\documentclass[acmsmall,nonacm]{acmart}

\pdfoutput=1
\AtBeginDocument{%
  }

\setcopyright{acmlicensed}
\acmJournal{PACMMOD}
\acmYear{2024} \acmVolume{2} \acmNumber{2 (PODS)} \acmArticle{91} \acmMonth{5}
\acmDOI{10.1145/3651592}

\usepackage[bibliography=common]{apxproof}

\usepackage{amsmath}
\usepackage{amsthm}
\usepackage{xcolor}
\usepackage{adjustbox}
\usepackage{calc}
\usepackage{tikz}
\usetikzlibrary{decorations.pathreplacing,angles,quotes,positioning,shapes.misc,decorations.text}
\usetikzlibrary{arrows.meta}
\tikzset{
  >={To[length=2.5pt]}
  }
\usepackage{subcaption}
\usepackage{xspace}
\usepackage{url}

\usepackage{wasysym}

\usepackage[normalem]{ulem}

\usepackage[ruled,vlined,noend]{algorithm2e}
\usepackage{listings}
\usepackage{caption}

\usepackage{hyperref}

\usepackage{verbatimbox}

\renewcommand{\epsilon}{\varepsilon}

\newcommand{\nop}[1]{}

\newcommand{\action}{\ensuremath{\text{action}}}
\newcommand{\object}{\ensuremath{\text{object}}}

\newcommand{\mydef}{\mathrel{:=}}
\newcommand{\ie}{\textit{i.e.}}
\newcommand{\eg}{\textit{e.g.}}

\newcommand{\trans}[1][i]{T_{#1}}

\newcommand{\transset}{{\mathcal{T}}}
\newcommand{\tfirst}{\textit{first}}
\newcommand{\objects}{\mathbf{Obj}}

\newcommand{\schedule}{s}
\newcommand{\schop}[1][\schedule]{O_{#1}}
\newcommand{\schord}[1][\schedule]{\leq_{#1}}
\newcommand{\schords}[1][\schedule]{<_{#1}}
\newcommand{\schvord}[1][\schedule]{\ll_{#1}}
\newcommand{\schvf}[1][\schedule]{v_{#1}}
\newcommand{\sstart}{\textit{op}_0}

\newcommand{\dependson}[3][\schedule]{#2 \rightarrow_{#1} #3}

\newcommand{\lastv}[2][\schedule]{\textit{last}_{\schvord[#1]}(#2)}

\newcommand{\isolationlevel}{\text{$\mathcal{I}$}}
\newcommand{\mvrc}{\text{RC}\xspace}
\newcommand{\si}{\text{SI}\xspace}
\newcommand{\ssi}{\text{SSI}\xspace}
\newcommand{\MVRC}{\mvrc}
\newcommand{\rc}{\mvrc}
\newcommand{\alloc}{\mathcal{A}}

\newcommand{\tmpvs}{\text{VS}}

\newcommand{\seg}[1]{SeG(#1)}
\newcommand{\cg}[1]{\seg{#1}}

\newcommand{\prefix}[2]{{\normalfont\textsf{prefix}}_{#2}(#1)}

\newcommand{\postfix}[2]{{\normalfont\textsf{postfix}}_{#2}(#1)}

\newcommand{\readun}{{\normalfont\textsc{read uncommitted}}\xspace}
\newcommand{\readcom}{{\normalfont\textsc{read committed}}\xspace}

\newcommand{\snapshot}{{\normalfont\textsc{snapshot isolation}}\xspace}

\newcommand{\isopostgres}{\ensuremath{\{\mvrc,\allowbreak\si,\allowbreak\ssi\}}\xspace}

\newcommand{\LOGSPACE}{{\sc logspace}\xspace}

\newcommand{\NP}{{\sc np}\xspace}
\newcommand{\coNP}{{\rm co{\sc np}}\xspace}

\newcounter{conditioncounter}
\newcommand{\condition}[2][C]{%
    \renewcommand{\theconditioncounter}{#1\arabic{conditioncounter}}
    \refstepcounter{conditioncounter}%
    \noindent
    \begin{minipage}[t]{1.0cm}\vspace{0pt}
        \center(\theconditioncounter)
    \end{minipage}
    \begin{minipage}[t]{
            \columnwidth-1.2cm}\vspace{0pt}#2\end{minipage}\par
    \medskip
}

\newcommand{\myR}{\ensuremath{\mathtt{R}}}
\newcommand{\myW}{\ensuremath{\mathtt{W}}}

\newcommand{\R}[2][i]{\myR_{#1}\mathtt{[#2]}}
\newcommand{\W}[2][i]{\myW_{#1}\mathtt{[#2]}}

\newcommand{\CT}[1][i]{\mathtt{C}_{#1}}

\newcommand{\x}{\mathtt{t}}
\newcommand{\y}{\mathtt{v}}
\newcommand{\z}{\mathtt{q}}

\usetikzlibrary{tikzmark,scopes,decorations.pathreplacing}

\begin{document}

\title[When View- and Conflict-Robustness Coincide for Multiversion Concurrency Control]{When View- and Conflict-Robustness Coincide \\for Multiversion Concurrency Control}

\author{Brecht Vandevoort}
\email{brecht.vandevoort@uhasselt.be}
\orcid{0000-0001-7212-4625}
\affiliation{%
    \institution{UHasselt, Data Science Institute, ACSL}
    \city{Diepenbeek}
    \country{Belgium}
}

\author{Bas Ketsman}
\email{bas.ketsman@vub.be}
\orcid{0000-0002-4032-0709}
\affiliation{%
    \institution{Vrije Universiteit Brussel}
    \city{Brussels}
    \country{Belgium}
}

\author{Frank Neven}
\email{frank.neven@uhasselt.be}
\orcid{0000-0002-7143-1903}
\affiliation{%
    \institution{UHasselt, Data Science Institute, ACSL}
    \city{Diepenbeek}
    \country{Belgium}
}



\begin{abstract}

A DBMS allows trading consistency for efficiency through the allocation of isolation levels that are strictly weaker than serializability. The \emph{robustness problem} asks whether, for a given set of transactions and a given allocation of isolation levels, every possible interleaved execution of those transactions that is allowed under the provided allocation, is always \emph{safe}. In the literature, safe is interpreted as conflict-serializable (to which we refer here as conflict-robustness). In this paper, we study the view-robustness problem, interpreting safe as view-serializable. View-serializability is a more permissive notion that allows for a greater number of schedules to be serializable and aligns more closely with the intuitive understanding of what it means for a database to be consistent. However, view-serializability is more complex to analyze (e.g., conflict-serializability can be decided in polynomial time whereas deciding view-serializability is NP-complete). While  conflict-robustness implies view-robustness, the converse does not hold in general. In this paper, we provide a sufficient condition for isolation levels guaranteeing that conflict- and view-robustness coincide and show that this condition is satisfied by the isolation levels occurring in Postgres and Oracle: read committed (RC), snapshot isolation (SI) and serializable snapshot isolation (SSI). It hence follows that for these systems, widening from conflict- to view-serializability does not allow for more sets of transactions to become robust. Interestingly, the complexity of deciding serializability within these isolation levels is still quite different. Indeed, deciding conflict-serializability for schedules allowed under RC and SI remains in polynomial time, while we show that deciding view-serializability within these isolation levels remains NP-complete.

\end{abstract}





\maketitle


\section{Introduction}
\label{sec:intro}

The holy grail in concurrency control is serializability. This notion  guarantees the highest level of isolation between transactions, ensuring that the results of a transaction remain invisible to other transactions until it is committed. Additionally, serializability offers a simple and intuitive model for programmers enabling them to focus exclusively on the correctness of individual transactions independent of any concurrent transactions. Formally, an interleaving or schedule of transactions is \emph{serializable} when it is equivalent to a serial execution of those transactions. The definition of serializability crucially depends on the chosen notion of equivalence, that is, when two schedules of transactions are considered to be equivalent. The most prevalent notion in the concurrency literature is that of \emph{conflict-equivalence} which requires that the ordering of conflicting operations in both schedules is preserved.\footnote{Two operations are conflicting when they access the same object in the database and at least one of them is a write operation.} A schedule then is \emph{conflict-serializable} when it is conflict-equivalent to a serial schedule. 

A more permissive notion of equivalence is that of \emph{view-equiv\-a\-lence} which requires that corresponding read operations in the two schedules see the same values and that the execution of both schedules leaves the database in the same state. 
\emph{View-serializability} is then based on view-equivalence and aligns more closely with the intuitive understanding of what it means for a database to be consistent as it focuses on the outcome rather than on the order of specific interactions between transactions. In addition, view-serializability is more broadly applicable than conflict-se\-ri\-al\-izability in that it allows strictly more schedules to be serializable. 
Given the many advantages and high desirability of serializable schedules, it might seem counterintuitive that in practice only the more stringent notion of conflict-serializability is considered. The main reason, however, is that reasoning about conflict-se\-ri\-al\-izability is much more intuitive than reasoning about view-serializability. Indeed, conflict-serializability can be characterized through the absence of cycles in a so-called conflict-graph which in turn gives rise to a very natural and efficient polynomial time decision algorithm. Moreover, the idea of defining admissible schedules in terms of the absence of cycles in a graph structure naturally extends to the definition of isolation levels where additional requirements are enforced on the type of edges in cycles (see, e.g., \cite{adya99}). In strong contrast, however, deciding view-serializability is NP-complete~\cite{DBLP:books/cs/Papadimitriou86}. Conflict-serializability can therefore be considered as a good approximation for view-serializability that is easier to implement and that can be the basis of concurrency control algorithms with acceptable performance. Next, we argue that in the context of robustness it does make sense to reconsider view-serializability.  

Many relational database systems offer a range of isolation levels, allowing users to trade in isolation guarantees for better performance. However, executing transactions concurrently at weaker degrees of isolation does carry some risk as it can result in specific anomalies. Nevertheless, there are situations when a group of transactions can be executed at an isolation level weaker than serializability without causing any errors. In this way, we get the higher isolation guarantees of serializability for the price of a weaker isolation level, which is typically implementable with a less expensive concurrency control mechanism. 
There is a famous example that is part of database folklore: the
TPC-C benchmark \cite{TPCC} is robust against Snapshot Isolation (SI), so there is no need to run a stronger, and more expensive, concurrency control algorithm than SI if the workload is just TPC-C. This has
played a role in the incorrect choice of SI as the general concurrency control algorithm for isolation level Serializable in Oracle and PostgreSQL
(before version 9.1, cf.\ 
\cite{DBLP:journals/tods/FeketeLOOS05}).

The property discussed above is called \emph{(conflict-)robustness}\footnote{Actually, robustness is the term used in the literature as it was always clear that robustness w.r.t.\  conflict-serializability was meant. We use the term conflict-robustness in this paper to distinguish it from view-robustness and just say robustness when the distinction does not matter.} \cite{DBLP:conf/concur/0002G16,DBLP:conf/pods/Fekete05,DBLP:journals/tods/FeketeLOOS05}: a set of transactions $\transset$ is called \emph{conflict-robust} against a given isolation level if every possible interleaving of the transactions in $\transset$ that is allowed under the specified isolation level is conflict-serializable. 
The robustness problem received quite a bit of attention in the literature, and we refer to Section~\ref{sec:relatedwork} for an extensive discussion of prior work. In~\cite{vldbpaper}, it was experimentally verified that, under high contention, workloads that are conflict-robust against Read Committed (\mvrc) can be evaluated faster under \mvrc compared to stronger isolation levels. This means that the stronger guarantee of serializability is obtained at the lower cost of evaluating under \mvrc. Unfortunately, not all workloads are conflict-robust against a weaker isolation level. A natural question is therefore whether such workloads are \emph{view-robust} which would imply that they could still be safely executed at the weaker isolation level despite not being conflict-robust. Here, safe refers to view-serializable which, as discussed above, better corresponds to the intuitive understanding of what it means for a database to be consistent. We discuss this next in more detail.

\emph{View-robustness} against an isolation level is defined in analogy to conflict-robustness where it is now required that every possible interleaving allowed by the isolation level must be \emph{view}-serializable. It readily follows that conflict-robustness implies view-robustness: when every allowed schedule is conflict-equivalent to a serial schedule, it is also view-equivalent to the same serial schedule as conflict-equivalence implies view-equivalence. However, we show that there are isolation levels for which view-robustness does not imply conflict-robustness (cf., Proposition~\ref{prop:rob:basicproperty}(\ref{basic:two})). That is, there are sets of transactions that are view-robust but not conflict-robust against those isolation levels, and view-robustness can therefore allow more sets of transactions to be safely executed.

In this paper, we study the view-robustness problem for the isolation levels occurring in PostgreSQL and Oracle: \mvrc, \si, and serializable snapshot isolation (\ssi)~\cite{DBLP:journals/tods/FeketeLOOS05,PortsGrittner2012}. 
We point out that the na\"ive algorithm which determines non-robustness by merely guessing an allowed schedule and testing that it is not view-serializable, is in the complexity class $\Sigma_2^P$. It is important to realize that even such high complexity does not necessarily rule out practical applicability. Indeed, detecting robustness is not an online problem that occurs while a concrete transaction schedule unfolds possibly involving thousands of transactions. The approach of \cite{vldbpaper,DBLP:conf/icdt/VandevoortK0N22,DBLP:conf/edbt/VandevoortK0N23} that we follow in this paper, is targeted at settings where transactions are generated by a handful  of transaction \emph{programs}, for instance, made available through an API. The TPC-C benchmark, for example, consists of five different transaction programs, from which an unlimited number of concrete transactions can be instantiated. Robustness then becomes a \emph{static} property that can be tested offline at API design time. When the small set of transaction programs passes the robustness test, the database isolation level can be set to the weaker isolation level for that API without fear of introducing anomalies. To further clarify the difference between transaction programs and mere transactions: a transaction is a sequence of read and write operations to concrete database objects while a transaction program is a parameterized transaction (possibly containing loops and conditions) that can be instantiated to form an unlimited number of concrete transactions. Previous work has shown that algorithms for deciding conflict-robustness for transaction programs use algorithms for deciding conflict-robustness for mere transactions as basic building blocks~\cite{vldbpaper,DBLP:conf/icdt/VandevoortK0N22}. It therefore makes sense to first study view-robustness for concrete transactions (as we do in this paper) in an effort to increase the amount of workloads that can be executed safely at weaker isolation levels.

Notwithstanding this inherent potential for practical applicability, we show the (at least to us) surprising result that for the isolation levels of PostgreSQL and Oracle, view-robustness \emph{always} implies conflict-robustness. The latter even extends to the setting of mixed allocations where transactions can be allocated to different isolation levels (as for instance considered in \cite{DBLP:conf/pods/VandevoortKN23,DBLP:conf/pods/Fekete05}). 
This means in particular, that for these systems, widening from conflict- to view-serializability does not allow for more sets of transactions to become robust. As a main technical tool, we identify a criterion (Condition~\ref{cond:rob:equivalence}) for isolation levels in terms of the existence of a counter-example schedule of a specific form that witnesses non-conflict-robustness. We then show that for classes of isolation levels that satisfy this condition, 
view-robustness always implies conflict-robustness (Theorem~\ref{theo:condition}) and prove that the class $\{\mvrc,\si,\ssi\}$ satisfies it
(Theorem~\ref{theo:rc-si-ssi-view-robust}).

While view- and conflict-robustness coincide for the class $\{\mvrc,\si,\allowbreak \ssi\}$, it is interesting to point out that conflict- and view-serializ\-abil\-i\-ty do not.  In fact, the complexity of the corresponding decision problems within these isolation levels is quite different.\footnote{That is, assuming P$\neq$NP.} Indeed, it readily follows that deciding conflict-serializability for schedules allowed under RC and SI remains in polynomial time,\footnote{For \ssi the problem is trivial as any schedule allowed under \ssi is both conflict- and view-serializable.} while we show that deciding view-serializability within these isolation levels remains NP-complete. 

In addition to the practical motivation outlined above to study view-robustness, the present paper can also be related to work done by 
Yannakakis~\cite{DBLP:journals/jacm/Yannakakis84} who showed that view- and conflict-robustness coincide for the class of isolation levels that allows all possible schedules.\footnote{That paper refers to the problems as view- and conflict-safety, and studies safety for more serializability notions.} The present paper can therefore be seen as an extension of that research line where it is obtained that 
(\emph{i}) view- and conflict-robustness do not coincide for all classes of isolation levels; and, (\emph{ii}) view- and conflict-robustness do coincide for the isolation levels in PostgreSQL and Oracle.

\paragraph{Outline.} This paper is further organized as follows. 
We introduce the necessary definitions in Section~\ref{sec:defs}. We discuss conflict- and view-robustness in Section~\ref{sec:robustness} and consider the complexity of deciding view-serializability in Section~\ref{sec:deciding:viewser}. Finally, we discuss related work in Section~\ref{sec:relatedwork}
and conclude in Section~\ref{sec:conclusions}.
All proofs are moved to the appendix.


\section{Definitions}
\label{sec:defs}

\subsection{Transactions and schedules}

We fix an infinite set of objects $\objects$. A \emph{transaction} $\trans[]$ over $\objects$ is a sequence of operations $o_1\cdots o_n$. To every operation $o$, we associate $\action(o)\in \{\myR,\myW,\CT[]\}$ and $\object(o)\in\objects$. We say that $o$ is a read, write or commit operation when $\action(o)$ equals $\myR$, $\myW$, and $\CT[]$, respectively, and that $o$ is an operation on $\object(o)$. 
In the sequel, we leave the set of objects $\objects$ implicit when it is clear from the context and just say transaction rather than transaction over $\objects$.
Formally, we model a transaction as a linear order $(\trans[],\leq_{\trans[]})$, where $\trans[]$ is the set of 
operations occurring in the transaction and $\leq_{\trans[]}$ encodes the ordering of the operations. As usual, we use $<_{\trans[]}$ to denote the strict ordering. For a transaction $\trans[]$, we use $\tfirst(\trans[])$ to refer to the first operation in $\trans[]$.

We introduce some notation to facilitate the exposition of examples. For an object $\x \in \objects$, we denote by $\R[]{\x}$ a \emph{read} operation on $\x$ and by $\W[]{\x}$ a \emph{write} operation on $\x$. We denote the special \emph{commit} operation simply by $\CT[]$. When considering a set $\transset$ of transactions, we assume that every transaction in the set has a unique id $i$ and write $\trans[i]$ to make this id explicit. Similarly, to distinguish the operations of different transactions, we add this id as a subscript to the operation. That is, we write $\W{\x}$ and $\R{\x}$ to denote a $\W[]{\x}$ and $\R[]{\x}$ occurring in transaction $\trans$; similarly $\CT[i]$ denotes the commit operation in transaction $\trans[i]$. To avoid ambiguity of this notation, in the literature it is commonly assumed that a transaction performs at most one write and one read operation per object (see, \eg\
\cite{DBLP:conf/sigmod/BerensonBGMOO95,DBLP:conf/pods/Fekete05}). \emph{We follow this convention only in examples and emphasize that all our results hold for the more general setting in which multiple writes and reads per object are allowed.}

A \emph{(multiversion) schedule} $\schedule$ over a set $\transset$ of transactions is a tuple $(\schop, \schord, \schvord, \schvf)$ where
\begin{itemize}
    \item $\schop$ is the set containing all operations of transactions in $\transset$ as well as a special operation $\sstart$ conceptually writing the initial versions of all existing objects,
    \item $\schord$ encodes the ordering of these operations,
    \item $\schvord$ is a \emph{version order} providing for each object $\x$ a total order over all write operations on $\x$ occurring in $\schedule$, and,
    \item $v_\schedule$ is a \emph{version function} mapping each read operation $a$ in $\schedule$ to either $\sstart$ or to a write operation in $\schedule$.
\end{itemize}
We require that $\sstart \schord a$ for every operation $a \in {\schop}$, {$\sstart \schvord a$ for every write operation $a \in {\schop}$}, {$a \schvord b$ iff $a <_{\trans[]} b$ for every pair of write operations $a, b$ occurring in a transaction $\trans[] \in \transset$ and writing to the same object,} and that $a <_{\trans[]} b$ implies $a \schords b$ for every $\trans[] \in \transset$ and every $a,b \in \trans[]$. 
We furthermore require that for every read operation $a$, $v_\schedule(a) \schords a$ and, if $v_\schedule(a) \neq \sstart$, then the operation $v_\schedule(a)$ is on the same object as $a$.
Intuitively, these requirements imply the following: $\sstart$ indicates the start of the schedule, the order of operations in $s$ is consistent with the order of operations in every transaction $\trans[]\in\transset$, and the version function maps each read operation $a$ to the operation that wrote the version observed by $a$.
If $v_\schedule(a)$ is $\sstart$, then $a$ observes the initial version of this object.
The version order $\schvord$ represents the order in which different versions of
an object are installed in the database.
For a pair of write operations on the
same object, this version order does not necessarily coincide with $\schord$.
For example, under \mvrc and \si the version order is based on the commit order instead.
If however these two write operations $a, b \in \schop$ occur in the same transaction $\trans[] \in \transset$, then we do require that these versions are installed in the order implied by the order of operations in $\transset$. That is, $a \schvord b$ iff $a <_{\trans[]} b$. See Figure~\ref{fig:ex:schedule} for an illustration of a schedule. {In this schedule, the read operations on $\x$ in $\trans[1]$ and $\trans[4]$ both read the initial version of $\x$ instead of the version written but not yet committed by $\trans[2]$. Furthermore, the read operation $\R[2]{\y}$ in $\trans[2]$ reads the initial version of $\y$ instead of the version written by $\trans[3]$, even though $\trans[3]$ commits before $\R[2]{\y}$.}

We say that a schedule $\schedule$ is a \emph{single-version schedule} if {$\schvord$ is compatible with $\schord$ and} every read operation always reads the last written version of the object. Formally, {for each pair of write operations $a$ and $b$ on the same object, $a \schvord b$ iff $a \schords b$, and} for every read operation $a$ there is no write operation $c$ on
the same object as $a$
with $v_\schedule(a) \schords c  \schords a$.
A single version schedule over a set of transactions $\transset$ is \emph{single-version
    serial} if its transactions are not interleaved with operations from other transactions. That is, for every $a,b,c \in {\schop}$ with $a <_{\schedule}
    b<_{\schedule} c$ and $a,c \in \trans[]$ implies $b \in \trans[]$ for every
$\trans[] \in \transset$.

    {The absence of aborts in our definition of schedule is consistent with the common assumption~\cite{DBLP:conf/pods/Fekete05,DBLP:conf/concur/0002G16} that an underlying recovery mechanism will rollback aborted transactions. We only consider isolation levels that only read committed versions. Therefore there will never be cascading aborts.}

\subsection{Conflict-Serializability}
\label{sec:ser}

Let $a_j$ and $b_i$ be two operations on the same object $\x$ from different transactions $\trans[j]$ and $\trans[i]$ in a set of transactions $\transset$. We then say that $b_i$ is \emph{conflicting} with $a_j$ if:
\begin{itemize}
    \item \emph{(ww-conflict)} $b_i = \W[i]{\x}$ and $a_j = \W[j]{\x}$; or,
    \item \emph{(wr-conflict)} $b_i = \W[i]{\x}$ and $a_j = \R[j]{\x}$; or,
    \item \emph{(rw-conflict)} $b_i = \R[i]{\x}$ and $a_j = \W[j]{\x}$.
\end{itemize}
In this case, we also say that $b_i$ and $a_j$ are conflicting operations.
Furthermore, commit operations and the special operation $\sstart$ never conflict with any other operation.
When $b_i$ and $a_j$ are conflicting operations in $\transset$, we say that $a_j$ \emph{depends on} $b_i$ in a schedule $\schedule$ over $\transset$, denoted $\dependson{b_i}{a_j}$ if:
\begin{itemize}
    \item \emph{(ww-dependency)} {{$b_i$ is ww-conflicting with $a_j$} and $b_i \schvord a_j$}; or, 
    \item \emph{(wr-dependency)} {{$b_i$ is wr-conflicting with $a_j$} and $b_i = \schvf(a_j)$ or $b_i \schvord \schvf(a_j)$}; or,
    \item \emph{(rw-antidependency)} {{$b_i$ is rw-conflicting with $a_j$} and\\ $\schvf(b_i) \schvord a_j$.}
\end{itemize}

Intuitively, a ww-dependency from $b_i$ to $a_j$ implies that $a_j$ writes a version of an object that is installed after the version written by $b_i$.
A wr-dependency from $b_i$ to $a_j$ implies that $b_i$ either writes the version observed by $a_j$, or it writes a version that is installed before the version observed by $a_j$.
A rw-antidependency from $b_i$ to $a_j$ implies that $b_i$ observes a version installed before the version written by $a_j$.
{For example, the dependencies $\dependson[\schedule_1]{\W[2]{\x}}{\W[4]{\x}}$, $\dependson[\schedule_1]{\W[3]{\y}}{\R[4]{\y}}$ and $\dependson[\schedule_1]{\R[4]{\x}}{\W[2]{\x}}$ are respectively a ww-dependency, a wr-dependency and a rw-antidependency in schedule $\schedule_1$ presented in Figure~\ref{fig:ex:schedule}.} 

Two schedules $\schedule$ and $\schedule'$ are \emph{conflict-equivalent} if they are over the same set $\transset$ of transactions and for every pair of conflicting operations $a_j$ and $b_i$, $\dependson[\schedule]{b_i}{a_j}$ iff $\dependson[\schedule']{b_i}{a_j}$.

\begin{definition}
    A schedule $\schedule$ is \emph{conflict-serializable} if it is conflict-equivalent to a single-version serial schedule.
\end{definition}

A {\emph{serialization graph}} $\cg{\schedule}$ for schedule $\schedule$ over a set of transactions $\transset$ is the graph whose nodes are the transactions in $\transset$ and where there is an edge from $T_i$ to $T_j$ if {$T_j$ has an operation $a_j$ that depends on an operation $b_i$ in $T_i$, thus with $\dependson{b_i}{a_j}$.}

\begin{theorem}[implied by \cite{DBLP:conf/icde/AdyaLO00}]\label{theo:not-conflict-serializable}
    A schedule $\schedule$ is conflict-serializable iff\/ $\cg{\schedule}$ is acyclic.
\end{theorem}

{Figure~\ref{fig:ex:sergraph} visualizes the serialization graph $\cg{\schedule_1}$ for the schedule $\schedule_1$ in Figure~\ref{fig:ex:schedule}. Since $\cg{\schedule_1}$ is not acyclic, $\schedule_1$ is not conflict-serializable.}

\begin{figure}
    \begin{center}
        \begin{tikzpicture}[remember picture]
            \node[anchor=west](node0) at (0,0.5) {$\sstart$};
            \node[anchor=west] at (0,0)
            {$\phantom{\sstart\,\W[2]{\x}\,\R[4]{y}\,
                        \W[3]{\x}\,\CT[3]}\, \subnode{node11}{\R[1]{\x}}\, \CT[1]\, \phantom{\R[2]{\y}\, \W[2]{\z}\, \CT[2]\, \W[3]{\z}\, \CT[3]}$};
            \node[anchor=west] at (0,-0.5) {$\phantom{\sstart}\,
                    \subnode{node21}{\W[2]{\x}}\,\phantom{\R[4]{\y}\, \W[3]{\y}\, \CT[3]\, \R[1]{\x}\, \CT[1]}\,
                    \subnode{node22}{\R[2]{\y}}\, \CT[2]$};
            \node[anchor=west] at (0,-1)
            {$\phantom{\sstart\,\W[2]{\x}\,\R[4]{\y}}\,
                    \subnode{node31}{\W[3]{\y}}\,
                    \CT[3]$};
            \node[anchor=west] at (0,-1.5)
            {$\phantom{\sstart\,\W[2]{\x}}\,\subnode{node41}{\R[4]{\x}}\,
                    \phantom{\W[3]{\y}\,
                        \CT[3]\, \R[1]{\x}\, \CT[1]\, \R[2]{\y}\, \CT[2]}\,
                    \subnode{node42}{\W[4]{\x}}\,\subnode{node43}{\R[4]{\y}}\,
                    \CT[4]$};
            
            \draw[->,solid,out=180,in=0] (node11) to (node0);
            
            \draw[->,solid,in=-15,out=180] (node22) to (node0);
            
            \draw[->,double,solid,in=180,out=-90] (node0) to (node21);
            
            \draw[->,solid,in=-135,out=180] (node41) to (node0);
            
            \draw[->,double,solid,out=-30,in=135] (node0) to (node31);
            
            \draw[->,solid,double,in=180,out=-35] (node21) to (node42);
            
            \draw[->,solid,out=160,in=20] (node43) to (node31);
            \node at (-0.5,0) {$\trans[1]:$};
            \node at (-0.5,-0.5) {$\trans[2]:$};
            \node at (-0.5,-1) {$\trans[3]:$};
            \node at (-0.5,-1.5) {$\trans[4]:$};
        \end{tikzpicture}

    \end{center}

    \caption{\label{fig:ex:schedule}
        A schedule $\schedule_1$ with $\schvf[\schedule_1]$ (single lines) and $\schvord[\schedule_1]$ (double lines) represented through
        arrows.
    }

\end{figure}
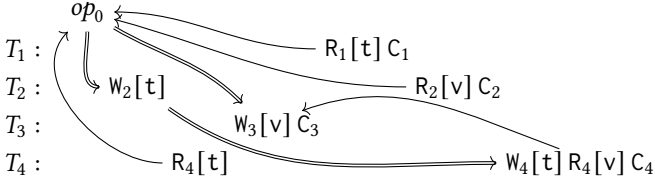

\subsection{View-Serializability}
\label{sec:vser}

For a schedule $\schedule$ and object $\x$ occurring in a write operation in a transaction in $s$, let $\lastv[\schedule]{\x}$ denote the write operation installing the last version of $\x$ in $\schedule$. That is, there is no write operation $b_i$ on $\x$ in $\schedule$ with $\lastv[\schedule]{\x} \schvord b_i$.
Two schedules $\schedule$ and $\schedule'$ are \emph{view-equivalent} if they are over the same set $\transset$ of transactions and for every read operation $b_i$ in $\transset$, $\schvf[\schedule](b_i) = \schvf[\schedule'](b_i)$ and $\lastv[\schedule]{\x} = \lastv[\schedule']{\x}$ for each 
$\x$ occurring in a write operation in a transaction in $s$. In other words, view-equivalence requires the two schedules to observe identical versions for each read operation and to install the same last versions for each object.

\begin{definition}
    A schedule $\schedule$ is \emph{view-serializable} if it is view-equivalent to a single-version serial schedule.
\end{definition}

\begin{theorem}[\cite{DBLP:books/cs/Papadimitriou86}]
    \label{theo:vs:npcomplete}
    Deciding whether a schedule $\schedule$ is view serializable is \NP-complete, even if\/ $\schedule$ is restricted to single-version schedules.
\end{theorem}

The following Theorem extends a well-known result~\cite{DBLP:books/cs/Papadimitriou86} for single-version schedules towards multiversion schedules:
\begin{theorem}
    \label{theo:cs:implies:vs}
    If a schedule $\schedule$ is conflict-serializable, then it is view-serializable.
\end{theorem}

\begin{toappendix}
    \subsection{Proof for Theorem~\ref{theo:cs:implies:vs}}
\begin{proof}
    We argue that if two schedules $\schedule$ and $\schedule'$ over a set of transactions $\transset$ are conflict-equivalent, then they are view-equivalent. The desired result then follows immediately.
    Notice that since $\schedule$ and $\schedule'$ are conflict-equivalent, we have $b_i \schvord a_j$ iff $b_i \schvord[\schedule'] a_j$ for each pair of conflicting write operations $a_j$ and $b_i$ occurring in $\transset$. Indeed, otherwise $a_j$ would depend on $b_i$ in one schedule, whereas $b_i$ would depend on $a_j$ in the other schedule. As an immediate consequence, we have that $\lastv[\schedule]{\x} = \lastv[\schedule']{\x}$ for each 
    object $\x$ occurring in a write operation in $\transset$.

    It remains to argue that $\schvf[\schedule](b_i) = \schvf[\schedule'](b_i)$ for each read operation $b_i$ occurring in $\transset$. Assume towards a contradiction that there is a read operation $b_i$ in $\transset$ with $\schvf[\schedule](b_i) \neq \schvf[\schedule'](b_i)$, and let $c = \schvf[\schedule](b_i)$ and $c' = \schvf[\schedule'](b_i)$. We first consider the case where $c \schvord[\schedule] c'$.
    Then, $\dependson[\schedule]{b_i}{c'}$ but $\dependson[\schedule']{c'}{b_i}$.
    For the case where $c' \schvord[\schedule] c$, we recall that $c' \schvord[\schedule'] c$ must hold as well. Analogous to the previous case, it now follows that $\dependson[\schedule']{b_i}{c}$ but $\dependson[\schedule]{c}{b_i}$.
    In both cases, $\schedule$ and $\schedule'$ are not conflict-equivalent, leading to the intended contradiction.
\end{proof}
\end{toappendix}

The opposite direction does not hold as the next example shows.

\begin{example}
    \label{ex:vs:not:cs}  
    Consider the schedule $\schedule_2$ over a set of three transactions $\transset = \{\trans[1], \trans[2], \trans[3]\}$ visualized in Figure~\ref{fig:ex:viewser}. This schedule is not conflict-serializable, witnessed by the cycle in $\seg{\schedule_2}$ given in Figure~\ref{fig:ex:sergraph}.
    However, $\schedule_2$ is view-serializable, as it is view equivalent to the single-version serial schedule $\schedule': \trans[1] \cdot \trans[2] \cdot \trans[3]$. Notice in particular that $\schvf[\schedule_2](\R[1]{\x}) = \schvf[\schedule'](\R[1]{\x}) =  \sstart$ and that $\trans[3]$ installs the last version of objects $\x$ and $\y$ in both schedules.
    Notice that view-serializability does not impose any restrictions on the ordering of $\W[1]{\x}$ and $\W[2]{\x}$ as $\x$ is not read by $T_3$.
    \hfill $\Box$
\end{example}

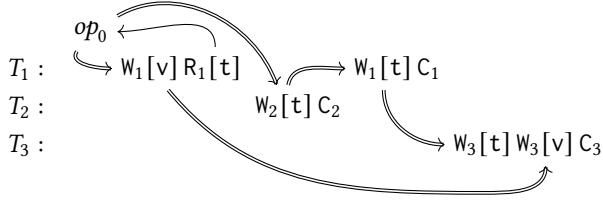
\begin{figure}
    \begin{center}
        \begin{tikzpicture}[remember picture]
            \node[anchor=west](s2node0) at (0,0.5) {$\sstart$};
            \node[anchor=west] at (0,0)
            {
                $
                \phantom{\sstart\,}\,\,
                \subnode{s2node11}{\W[1]{\y}}\,\subnode{s2node12}{\R[1]{\x}\,}\,\,
                \phantom{\W[2]{\x}\,\CT[2]\,}\,\,
                \subnode{s2node13}{\W[1]{\x}}\,{\CT[1]\,}\,\,
                \phantom{\W[3]{\x}\,\W[3]{\y}\,\CT[3]}\,\,
                $
            };
            \node[anchor=west] at (0,-0.5)
            {
                $
                \phantom{\sstart\,}\,\,
                \phantom{\W[1]{\y}\,\R[1]{x}\,}\,\,
                \subnode{s2node21}{\W[2]{\x}}{\,\CT[2]\,}\,\,
                \phantom{\W[1]{\x}\,\CT[1]\,}\,\,
                \phantom{\W[3]{\x}\,\W[3]{\y}\,\CT[3]}\,\,
                $
            };
            \node[anchor=west] at (0,-1)
            {
                $
                \phantom{\sstart\,}\,\,
                \phantom{\W[1]{\y}\,\R[1]{x}\,}\,\,
                \phantom{\W[2]{\x}\,\CT[2]\,}\,\,
                \phantom{\W[1]{\x}\,\CT[1]\,}\,\,
                \subnode{s2node31}{\W[3]{\x}}\,\subnode{s2node32}{\W[3]{\y}}{\,\CT[3]}\,\,
                $
            };

            \draw[->,solid,out=90,in=0] (s2node12) to (s2node0);

            \draw[->,double,solid,in=180,out=-125] (s2node0) to (s2node11);
            \draw[->,double,solid,in=110,out=30] (s2node0) to (s2node21);
            \draw[->,double,solid,in=180,out=80] (s2node21) to (s2node13);
            
            \coordinate[below=10pt of s2node31] (s2x);
            \draw[->,double,solid] (s2node11) to[in=180,out=-50] (s2x) to[in=-90,out=0] (s2node32);
            
            \draw[->,double,solid,in=180,out=-90] (s2node13) to (s2node31);
            
            \node at (-0.5,0) {$\trans[1]:$};
            \node at (-0.5,-0.5) {$\trans[2]:$};
            \node at (-0.5,-1) {$\trans[3]:$};
        \end{tikzpicture}
    \end{center}
    \caption{\label{fig:ex:viewser}
    A schedule $\schedule_2$ with $\schvf[\schedule_2]$ (single lines) and $\schvord[\schedule_2]$ (double lines) represented through
    arrows.}
\end{figure}

\begin{figure}
    \begin{center}
        \begin{tikzpicture}
            \node[draw,circle] (T1) at (0,0) {$\trans[1]$};
            \node[draw,circle] (T2) at (1.5,0) {$\trans[2]$};
            \node[draw,circle] (T3) at (1.5,-1.5) {$\trans[3]$};
            \node[draw,circle] (T4) at (0,-1.5) {$\trans[4]$};
            \draw[->,bend left] (T1) to (T2);
            \draw[->,bend right] (T1) to (T4);
            \draw[->,bend left] (T2) to (T4);
            \draw[->,bend left] (T4) to (T2);
            \draw[->,bend left] (T2) to (T3);
            \draw[->,bend left] (T3) to (T4);

            \begin{scope}[xshift=4cm]
                
                \node[draw,circle] (s2T1) at (0,0) {$\trans[1]$};
                \node[draw,circle] (s2T2) at (1.8,0) {$\trans[2]$};
                \node[draw,circle] (s2T3) at (0.9,-1.5) {$\trans[3]$};
                \draw[->,bend right] (s2T1) to (s2T3);
                \draw[->,bend left] (s2T1) to (s2T2);
                \draw[->,bend left] (s2T2) to (s2T3);
                \draw[->,bend left] (s2T2) to (s2T1);
            \end{scope}
        \end{tikzpicture}
    \end{center}
    \caption{\label{fig:ex:sergraph}
        {Serialization graphs $\cg{\schedule_1}$ \emph{(left)} and $\cg{\schedule_2}$ \emph{(right)} for the schedules $\schedule_1$ and $\schedule_2$ presented in Figure~\ref{fig:ex:schedule} and Figure~\ref{fig:ex:viewser}.}
    }
\end{figure}
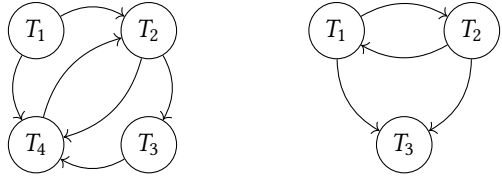

\subsection{Isolation levels}
\label{sec:isolevel}

Let $\Sigma_\transset$ be the set of all possible schedules over a set of transactions $\transset$. An \emph{allocation $\alloc$} for $\transset$ defines a set $\Sigma_\alloc \subseteq \Sigma_\transset$. A schedule $\schedule$ over $\transset$ is \emph{allowed under $\alloc$} if $\schedule \in \Sigma_\alloc$. We furthermore assume that $\alloc$ implies an allocation for each set of transactions $\transset' \subseteq \transset$, and write $\alloc[\transset']$ to denote the allocation obtained by restricting $\alloc$ to the transactions in $\transset'$.
Typically, the set of allowed schedules is defined through the isolation level assigned to each transaction. To this end, let $\isolationlevel$ be a class of isolation levels.
An \emph{$\isolationlevel$-allocation $\alloc$} for a set of transactions $\transset$ is an allocation to which we relate a function $f_\alloc$ mapping each transaction $\trans[] \in \transset$ onto an isolation level $f_\alloc(\trans[]) \in \isolationlevel$. As a slight abuse of notation, we will frequently denote $f_\alloc(\trans[])$ by $\alloc(\trans[])$. Intuitively, $\Sigma_\alloc$ then contains all schedules that can be obtained by executing the transactions in $\transset$ under the isolation levels prescribed by $\alloc$.

In this paper, we consider the following isolation levels: read committed (\mvrc), snapshot isolation (\si), and serializable snapshot isolation (\ssi).
Before we define when a schedule is allowed under a $\isopostgres$-allocation, we 
introduce some necessary terminology.
Some of these notions are illustrated in Example~\ref{ex:schedule} below.

Let $\schedule$ be a schedule for a set $\transset$ of transactions.
Two transactions $\trans[i], \trans[j] \in \transset$ are said to be \emph{concurrent} in $\schedule$ when their execution overlaps. That is,
if $\tfirst(\trans[i]) \schords \CT[j]$ and $\tfirst(\trans[j]) \schords \CT[i]$.
We say that a write operation $\W[j]{\x}$ in a transaction $\trans[j] \in \transset$ \emph{respects the commit order of $\schedule$} if the version of $\x$ written by $\trans[j]$ is installed after all versions of $\x$ installed by transactions committing before $\trans[j]$ commits, but before all versions of $\x$ installed by transactions committing after $\trans[j]$ commits. More formally, if for every write operation $\W[i]{\x}$ in a transaction $\trans[i] \in \transset$ different from $\trans[j]$ we have $\W[j]{\x} \schvord \W[i]{\x}$ iff $\CT[j] \schords \CT[i]$.   
We next define when a read operation $a\in \trans[]$ reads the last committed version relative to a specific operation. For \rc this operation is $a$ itself while for \si this operation is $\tfirst(\trans[])$.
{Intuitively, these definitions enforce that read operations in transactions allowed under \rc act as if they observe a snapshot taken right before the read operation itself, while under \si they observe a snapshot taken right before the first operation of the transaction.}
A read operation $\R[j]{\x}$ in a transaction $\trans[j] \in \transset$ is \emph{read-last-committed in $\schedule$ relative to an operation $a_j \in \trans[j]$} (not necessarily different from $\R[j]{\x}$) if the following holds:
\begin{itemize}
    \item $\schvf(\R[j]{\x}) = \sstart$ or $\schvf(\R[j]{\x}) \in \trans[i]$ with $\CT[i] \schords a_j$ for some $\trans[i] \in \transset$; and
    \item there is no write operation $\W[k]{\x} \in \schop$ with $\CT[k] \schords a_j$ and $\schvf(\R[j]{\x}) \schvord \W[k]{\x}$.
\end{itemize}
The first condition says that $\R[j]{\x}$ either reads the initial version or a committed version, while the second condition states that
$\R[j]{\x}$ observes the most recently committed version of $\x$ (according to $\schvord$).
A transaction $\trans[j] \in \transset$ \emph{exhibits a concurrent write in $\schedule$} if there {is another transaction $\trans[i]\in\transset$ and there} are two write operations $b_i$ and $a_j$ in $\schedule$ on the same object with $b_i \in \trans[i]$, $a_j \in \trans[j]$ and $\trans[i] \neq \trans[j]$ such that $b_i \schords a_j$ and $\tfirst(\trans[j]) \schords \CT[i]$. That is, transaction $\trans[j]$ writes to an object that has been modified earlier by a concurrent transaction $\trans[i]$.

A transaction $\trans[j] \in \transset$ \emph{exhibits a dirty write in $\schedule$} if there are two write operations $b_i$ and $a_j$ in $\schedule$ with $b_i \in \trans[i]$, $a_j \in \trans[j]$ and $\trans[i] \neq \trans[j]$ such that $b_i \schords a_j \schords \CT[i]$.
That is, transaction $\trans[j]$ writes to an object that has been modified earlier by $\trans[i]$, but $\trans[i]$ has not yet issued a commit.
Notice that by definition a transaction exhibiting a dirty write always exhibits a concurrent write.
{Transaction $\trans[4]$ in Figure~\ref{fig:ex:schedule} exhibits a concurrent write, since it writes to $\x$, which has been modified earlier by a concurrent transaction $\trans[2]$. However, $\trans[4]$ does not exhibit a dirty write, since $\trans[2]$ has already committed before $\trans[4]$ writes to $\x$.}

\begin{definition}
    Let $\schedule$ be a schedule over a set of transactions $\transset$.
    A transaction $\trans[i] \in \transset$ is \emph{allowed under isolation level read committed (\mvrc) in $\schedule$} if:
    \begin{itemize}
        \item each write operation in $\trans[i]$ respects the commit order of $\schedule$;
        \item each read operation $b_i \in \trans[i]$ is read-last-committed in $\schedule$ relative to $b_i$; and
        \item $\trans[i]$ does not exhibit dirty writes in $\schedule$.
    \end{itemize}
    A transaction $\trans[i] \in \transset$ is \emph{allowed under isolation level snapshot isolation (\si) in $\schedule$} if:
    \begin{itemize}
        \item each write operation in $\trans[i]$ respects the commit order of $\schedule$;
        \item each read operation in $\trans[i]$ is read-last-committed in $\schedule$ relative to $\tfirst(\trans[i])$; and
        \item $\trans[i]$ does not exhibit concurrent writes in $\schedule$.
    \end{itemize}
\end{definition}

We then say that the schedule $s$ is allowed under \mvrc (respectively, \si) if every transaction is allowed under \mvrc (respectively, \si) in $s$. The latter definitions correspond to the ones in the literature (see, e.g., \cite{DBLP:conf/pods/Fekete05,vldbpaper}).
{We emphasize that our definition of \rc is based on concrete implementations over multiversion databases, found in e.g. PostgreSQL, and should therefore not be confused with different interpretations of the term Read Committed, such as lock-based implementations~\cite{DBLP:conf/sigmod/BerensonBGMOO95} or more abstract specifications covering a wider range of concrete implementations (see, e.g.,~\cite{DBLP:conf/icde/AdyaLO00}). In particular, abstract specifications such as~\cite{DBLP:conf/icde/AdyaLO00} do not require the read-last-committed property, thereby facilitating implementations in distributed settings, where read operations are allowed to observe outdated versions. When studying robustness, such a broad specification of \rc is not desirable, since it allows for a wide range of schedules that are not conflict-serializable.
We furthermore point out that our definitions of \rc and \si are not strictly weaker forms of conflict-serializability or view-serializability. That is, a conflict-serializable (respectively, view-serializable) schedule is not necessarily allowed under \rc and \si as we discuss further in Section~\ref{sec:deciding:viewser}.
}

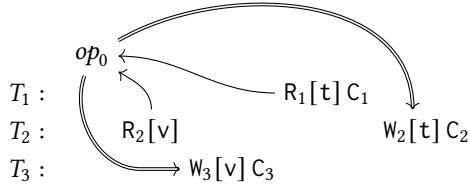
\begin{figure}
    \begin{center}
        \begin{tikzpicture}[remember picture]
            \node[anchor=west](sdnode0) at (0,0.5) {$\sstart$};
            \node[anchor=west] at (0,0)
            {
                $
                \phantom{\sstart\,}\,\,
                \phantom{{\R[2]{\y}}}\,\,
                \phantom{{\W[3]{\y}}\,\CT[3]}\,\,
                \subnode{sdnode11}{\R[1]{\x}}\,{\CT[1]\,}\,\,
                \phantom{{\W[2]{\x}}\,\CT[2]}\,\,
                $
            };
            \node[anchor=west] at (0,-0.5)
            {
                $
                \phantom{\sstart\,}\,\,
                \subnode{sdnode21}{\R[2]{\y}}\,\,
                \phantom{{\W[3]{\y}}\,\CT[3]}\,\,
                \phantom{{\R[1]{\x}}\,{\CT[1]\,}}\,\,
                \subnode{sdnode22}{\W[2]{\x}}\,\CT[2]\,\,
                $
            };
            \node[anchor=west] at (0,-1)
            {
                $
                \phantom{\sstart\,}\,\,
                \phantom{{\R[2]{\y}}}\,\,
                \subnode{sdnode31}{\W[3]{\y}}\,\CT[3]\,\,
                \phantom{{\R[1]{\x}}\,{\CT[1]\,}}\,\,
                \phantom{{\W[2]{\x}}\,\CT[2]}\,\,
                $
            };

            \draw[->,solid,out=180,in=0] (sdnode11) to (sdnode0);
            \draw[->,solid,out=90,in=-30] (sdnode21) to (sdnode0);

            \draw[->,double,solid,in=90,out=30] (sdnode0) to (sdnode22);
            \coordinate[left=15pt of sdnode31] (sdx);
            \draw[->,double,solid,in=180,out=-110] (sdnode0) to[in=180,out=-110] (sdx) to[in=180,out=0] (sdnode31);

            \node at (-0.5,0) {$\trans[1]:$};
            \node at (-0.5,-0.5) {$\trans[2]:$};
            \node at (-0.5,-1) {$\trans[3]:$};
        \end{tikzpicture}
    \end{center}
    \caption{\label{fig:ex:dangerousstructure}
    A schedule $\schedule$ with $\schvf$ (single lines) and $\schvord$ (double lines) represented through
    arrows. The three transactions form a dangerous structure $\trans[1] \rightarrow \trans[2] \rightarrow \trans[3]$.}
\end{figure}

While \mvrc and \si are defined on the granularity of a single transaction,
\ssi enforces a global condition on the schedule as a whole.
For this, recall the concept of dangerous structures from~\cite{DBLP:conf/sigmod/CahillRF08}:
three transactions $\trans[1],\trans[2],\trans[3] \in \transset$ (where $\trans[1]$ and $\trans[3]$ are not necessarily different) form a \emph{dangerous structure} $\trans[1] \rightarrow \trans[2] \rightarrow \trans[3]$ in $\schedule$ if:
\begin{itemize}
    \item there is a rw-antidependency from $\trans[1]$ to $\trans[2]$ and from $\trans[2]$ to $\trans[3]$ in $\schedule$;
    \item $\trans[1]$ and $\trans[2]$ are concurrent in $\schedule$;
    \item $\trans[2]$ and $\trans[3]$ are concurrent in $\schedule$;
    \item ${\CT[3] \schord \CT[1]}$ and $\CT[3] \schords \CT[2]$; and
    \item {if $\trans[1]$ only contains read operations, then $\CT[3] \schords \tfirst(\trans[1])$.}
\end{itemize}

An example of a dangerous structure is visualized in Figure~\ref{fig:ex:dangerousstructure}. Note that this definition of dangerous structures slightly extends the one in~\cite{DBLP:conf/sigmod/CahillRF08}, where it is not required for $\trans[3]$ to commit before $\trans[1]$ and $\trans[2]$. In the full version~\cite{DBLP:journals/tods/CahillRF09} of that paper, it is shown that such a structure can only lead to non-serializable schedules if $\trans[3]$ commits first.
{Furthermore, Ports and Grittner~\cite{PortsGrittner2012} show that if $\trans[1]$ is a read-only transaction, this structure can only lead to non-serializable behaviour if $\trans[3]$ commits before $\trans[1]$ starts.} Actual implementations of \ssi (e.g., PostgreSQL~\cite{PortsGrittner2012}) therefore include this optimization when monitoring for dangerous structures to reduce the number of aborts due to false positives.

We are now ready to define when a schedule is allowed under a  $\isopostgres$-allocation.

\begin{definition} \label{def:mixed:schedule}
    A schedule $\schedule$ over a set of transactions $\transset$ \emph{is allowed under an $\isopostgres$-allocation $\alloc$} over $\transset$ if:
    \begin{itemize}
        \item for every transaction $\trans[i] \in \transset$ with $\alloc(\trans[i]) = \mvrc$, $\trans[i]$ is allowed under \mvrc;
        \item for every transaction $\trans[i] \in \transset$ with $\alloc(\trans[i]) \in \{\si,\ssi\}$, $\trans[i]$ is allowed under \si; and
        \item there is no dangerous structure $\trans[i] \rightarrow \trans[j] \rightarrow \trans[k]$ in $\schedule$ formed by three (not necessarily different) transactions $\trans[i], \trans[j], \trans[k] \in \{\trans[] \in \transset \mid \alloc(\trans[]) = \ssi\}$.
    \end{itemize}
\end{definition}

We illustrate some of the just introduced notions through an example.
\begin{example}\label{ex:schedule}

    Consider the schedule $\schedule_1$ in Figure~\ref{fig:ex:schedule}.
    Transaction $T_1$ is concurrent with $T_2$ and $T_4$, but not with
    $T_3$; all other transactions are pairwise concurrent with each other.
    The second read operation of $T_4$ is a read-last-committed relative to
    itself but not relative to the start of $T_4$.
        {The read operation of $\trans[2]$ is read-last-committed relative to the start of $\trans[2]$, but not relative to itself, so an allocation mapping $\trans[2]$ to \rc is not allowed.}
    All other read operations
        {are read-last-committed relative to both themselves and the start of the corresponding transaction.}
    None of the transactions exhibits a dirty write.
    Only transaction $T_4$ exhibits a concurrent write (witnessed by the write
    operation in $T_2$). Due to this, an allocation mapping $T_4$ on
    \si or \ssi is not allowed.
    The transactions $T1\to T2 \to T3$ form a dangerous structure, therefore
    an allocation mapping all three transactions $T_1,T_2,T_3$ on \ssi is not
    allowed.
    All other allocations, that is, mapping $T_4$ on \rc, {$\trans[2]$ on \si or \ssi} and at least one of
    $T_1,T_2,T_3$ on \rc or \si, is allowed. 
       \hfill $\Box$
\end{example}


\section{Conflict- and View-Robustness}
\label{sec:robustness}

\subsection{Definitions and basic properties}
\label{sec:robustness:definitions}

We define the robustness property~\cite{DBLP:conf/concur/0002G16} (also called \emph{acceptability} in~\cite{DBLP:conf/pods/Fekete05,DBLP:journals/tods/FeketeLOOS05}), which guarantees serializability for all schedules over a given set of transactions for a given allocation.

\begin{definition}[Robustness]
    \label{def:robustness}
    A set of transactions $\transset$ is \emph{view-robust} (respectively, \emph{conflict-robust}) against an {allocation $\alloc$} for $\transset$
    if every schedule over every $\transset' \subseteq \transset$ that is allowed under $\alloc[\transset']$ is
    view-serializable (respectively, conflict-serializable).
\end{definition}

It is important to note that the above definition demands serializability for every schedule over \emph{every subset of} $\transset$. {View- and conflict-robustness correspond to view- and conflict-safety as defined by Yannakakis~\cite{DBLP:journals/jacm/Yannakakis84}.}

We can define a less stringent version as follows: 
\begin{definition}
A set of transactions $\transset$ is \emph{exact view-robust} (respectively, \emph{exact conflict-robust}) against an {allocation $\alloc$} for $\transset$
if every schedule over $\transset$ that is allowed under $\alloc[\transset]$ is
view-serializable (respectively, conflict-serializable).
\end{definition}
Exact conflict-robustness is what is used in, e.g., \cite{DBLP:conf/pods/Fekete05,DBLP:conf/pods/VandevoortKN23}.\footnote{In these works exact robustness is just called robustness.} However, for conflict-robustness the distinction does not matter as exact conflict-robustness is the same as conflict-robustness. Indeed, a schedule $\schedule'$ over a subset $\transset' \subseteq \transset$ that is not conflict-serializable can always be extended to a schedule $\schedule$ over $\transset$ that is not conflict-serializable by appending the remaining transactions in $\transset$ (those not appearing in $\transset'$) to $\schedule'$ in a serial fashion.\footnote{We 
    assume that if $\schedule'$ is allowed under $\alloc[\transset']$ then the extended schedule $\schedule$ 
    is allowed under $\alloc$ as well.
    We formally prove this property in Lemma~\ref{lemma:nonconfser:extended} for the isolation levels considered in this paper.
    } 
Indeed, the cycle in the serialization graph cannot disappear by adding transactions to the schedule. For view-serializability, this argument does not work. In fact, the obtained 
schedule $\schedule$ can be view-serializable, even if $\schedule'$ is not view-serializable.  

As explained in the introduction, the relevance of robustness lies in its utility as a static property that is tested offline w.r.t.\ a small set of transaction programs. Here, any instantiation of any subset of such templates should be taken into account. Therefore,
view-robustness is the desirable property, not exact view-robustness. We next give an example 
of a set of transactions that is exact view-robust but not view-robust.

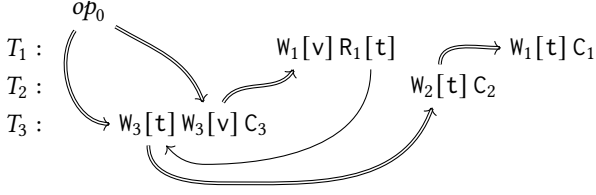
\begin{figure}
    \begin{center}
        \begin{tikzpicture}[remember picture]
            \node[anchor=west](s3node0) at (0,0.5) {$\sstart$};
            \node[anchor=west] at (0,0)
            {
                $
                \phantom{\sstart\,}\,\,
                \phantom{\W[3]{\x}\,\W[3]{\y}\,\CT[3]}\,\,
                \subnode{s3node11}{\W[1]{\y}}\,\subnode{s3node12}{\R[1]{\x}\,}\,\,
                \phantom{\W[2]{\x}\,\CT[2]\,}\,\,
                \subnode{s3node13}{\W[1]{\x}}\,{\CT[1]\,}\,\,
                $
            };
            \node[anchor=west] at (0,-0.5)
            {
                $
                \phantom{\sstart\,}\,\,
                \phantom{\W[3]{\x}\,\W[3]{\y}\,\CT[3]}\,\,
                \phantom{\W[1]{\y}\,\R[1]{x}\,}\,\,
                \subnode{s3node21}{\W[2]{\x}}{\,\CT[2]\,}\,\,
                \phantom{\W[1]{\x}\,\CT[1]\,}\,\,
                $
            };
            \node[anchor=west] at (0,-1)
            {
                $
                \phantom{\sstart\,}\,\,
                \subnode{s3node31}{\W[3]{\x}}\,\subnode{s3node32}{\W[3]{\y}}{\,\CT[3]}\,\,
                \phantom{\W[1]{\y}\,\R[1]{x}\,}\,\,
                \phantom{\W[2]{\x}\,\CT[2]\,}\,\,
                \phantom{\W[1]{\x}\,\CT[1]\,}\,\,
                $
            };

            \coordinate[below right=10pt of s3node31] (s3y);
            \draw[->,solid] (s3node12) to[out=-90,in=0] (s3y) to[in=-50,out=-180] (s3node31);

            \draw[->,double,solid,in=180,out=-125] (s3node0) to (s3node31);
            \draw[->,double,solid,in=110,out=-30] (s3node0) to (s3node32);
            \draw[->,double,solid,in=-120,out=60] (s3node32) to (s3node11);
            
            \coordinate[below right=20pt of s3node31] (s3x);
            \draw[->,double,solid] (s3node31) to[in=180,out=-90] (s3x) to[in=-110,out=0] (s3node21);
            
            \draw[->,double,solid,in=180,out=90] (s3node21) to (s3node13);
            
            \node at (-0.5,0) {$\trans[1]:$};
            \node at (-0.5,-0.5) {$\trans[2]:$};
            \node at (-0.5,-1) {$\trans[3]:$};
        \end{tikzpicture}
    \end{center}
    \caption{\label{fig:ex:viewser:prepend}
    A schedule $\schedule_4$ with $\schvf[\schedule_4]$ (single lines) and $\schvord[\schedule_4]$ (double lines) represented through
    arrows.}
\end{figure}

\begin{example}
    Reconsider the set of transactions $\transset = \{\trans[1], \trans[2], \trans[3]\}$ from Example~\ref{ex:vs:not:cs}.
    Let $\alloc$ be the $\isopostgres$-allocation for $\transset$ with $\alloc(\trans[i]) = \rc$ for each $\trans[i] \in \transset$.
    The schedule $\schedule_2$ over $\transset$ presented in Figure~\ref{fig:ex:viewser} is allowed under $\alloc$ but not conflict-serializable, thereby witnessing that $\transset$ is not (exact) conflict-robust against $\alloc$.
    
    However, $\transset$ is exact view-robust against $\alloc$ since every schedule $\schedule$ over $\transset$ allowed under $\alloc$ is view-serializable. Notice in particular that if $\schedule$ is allowed under $\alloc$, then $\trans[1]$ and $\trans[3]$ cannot be concurrent in $\schedule$, as this would always imply a dirty write. For every schedule $\schedule$ allowed under $\alloc$, we can therefore identify a view-equivalent single-version serial schedule $\schedule'$ over $\transset$, depending on $\schvf[\schedule](\R[1]{\x})$:
    \begin{itemize}
        \item if $\schvf[\schedule](\R[1]{\x}) = \sstart$, then $\schedule$ is view-equivalent to either $\trans[1] \cdot \trans[2] \cdot \trans[3]$ or $\trans[1] \cdot \trans[3] \cdot \trans[2]$, depending on whether $\trans[2]$ commits before $\trans[3]$ in $\schedule$;
        \item if $\schvf[\schedule](\R[1]{\x}) = \W[2]{\x}$, then $\schedule$ is view-equivalent to either $\trans[2] \cdot \trans[1] \cdot \trans[3]$ or $\trans[3] \cdot \trans[2] \cdot \trans[1]$, depending on whether $\trans[1]$ commits before $\trans[3]$ in $\schedule$; and
        \item if $\schvf[\schedule](\R[1]{\x}) = \W[3]{\x}$, then $\schedule$ is view-equivalent to either $\trans[3] \cdot \trans[1] \cdot \trans[2]$ or $\trans[2] \cdot \trans[3] \cdot \trans[1]$, depending on whether $\trans[1]$ commits before $\trans[2]$ in $\schedule$.
    \end{itemize}

    Even though $\transset$ is exact view-robust against $\alloc$, it is not view-robust against $\alloc$, as we will show next. Let $\schedule_3$ be the schedule over $\transset' = \{\trans[1], \trans[2]\}$ obtained by removing $\trans[3]$ from $\schedule_2$ in Figure~\ref{fig:ex:viewser}. Then, $\schedule_3$ is allowed under $\alloc[\transset']$ but not view-serializable. Indeed, $\trans[1]$ observes the initial version of object $\x$ and installs the last version of $\x$, and should therefore occur both before and after $\trans[2]$ in a view-equivalent single-version serial schedule, leading to the desired contradiction.

    {Notice in particular that $\schedule_3$ cannot be extended to a non-view-serializable schedule over $\transset$ by appending or prepending $\trans[3]$. Indeed, $\schedule_2$ in Figure~\ref{fig:ex:viewser} is the view-serializable schedule obtained by appending $\trans[3]$, whereas the schedule $\schedule_4$ obtained by prepending $\trans[3]$ is given in Figure~\ref{fig:ex:viewser:prepend}. As explained above, $\schedule_4$ is view-equivalent to the serial schedule $\trans[2] \cdot \trans[3] \cdot \trans[1]$. The crucial difference between $\schedule_3$ and $\schedule_4$ is that by prepending $\trans[3]$, transaction $\trans[1]$ no longer observes the initial version of object $\x$ but the version written by $\trans[3]$ instead. This allows $\trans[2]$ to be situated before $\trans[1]$ in a view-equivalent schedule, as long as $\trans[3]$ is situated between $\trans[2]$ and $\trans[1]$ in this view-equivalent schedule.}
    \hfill $\Box$
\end{example}

The next result shows that conflict-robustness always implies view-robustness but not vice-versa.
{
\begin{proposition} \label{prop:rob:basicproperty}
    \begin{enumerate}
        \item \label{basic:one} For every allocation $\alloc$ for a set of transactions $\transset$, if\/ $\transset$ is conflict-robust against $\alloc$, then $\transset$ is view-robust against $\alloc$.
    \item \label{basic:two} There is a set of transactions $\transset$ and an allocation $\alloc$ for $\transset$, such that $\transset$ is view-robust against $\alloc$ but not conflict-robust.  
\end{enumerate}
\end{proposition}
}

\begin{toappendix}
    \subsection{Proof for Proposition~\ref{prop:rob:basicproperty}}
\begin{proof} (1)
    The proof is by contraposition. Assume $\transset$ is not view-robust against $\alloc$. Then, there exists a schedule $\schedule$ over a set of transactions $\transset' \subseteq \transset$ that is allowed under $\alloc[\transset']$ but not view-serializable. By Theorem~\ref{theo:cs:implies:vs}, $\schedule$ is not conflict-serializable. Hence, $\transset$ is not conflict-robust against $\alloc$.

    \smallskip
    \noindent
    (2) 
    Consider the set of transactions $\transset$ from Example~\ref{ex:vs:not:cs}.
    Let $\alloc_\tmpvs$ be the allocation for $\transset$ that allows all schedules over subsets of $\transset$ that are view-serializable. That is, a schedule $\schedule$ over a subset $\transset' \subseteq \transset$ is allowed under $\alloc_\tmpvs[\transset']$ iff $\schedule$ is view-serializable. 
    Then, $\transset$ is trivially view-robust against $\alloc_\tmpvs$. However, $\transset$ is not conflict-robust against $\alloc_\tmpvs$ since the schedule $\schedule_2$ given in Figure~\ref{fig:ex:viewser} is allowed under $\alloc_\tmpvs$ but not conflict-serializable.      
\end{proof}
\end{toappendix}

\subsection{Generalized split schedules}
Towards identifying classes of isolation levels $\isolationlevel$ for which view-robustness and conflict-robustness coincide, we first introduce the notion of a \emph{generalized split schedule}.

In the next definition, we use $\prefix{\trans[]}{b}$ for an operation $b$ in $\trans[]$ to denote the restriction of $\trans[]$ to all operations that are before or equal to $b$ according to $\leq_{\trans[]}$. Similarly, we denote by $\postfix{\trans[]}{b}$ the restriction of $\trans[]$ to all operations that are strictly after $b$ according to $\leq_{\trans[]}$.

\begin{definition}[Generalized split schedule]\label{def:generalizedsplitschedule}
    Let $\transset = \{\trans[1], \trans[2], \ldots \trans[n]\}$ be a set of transactions.
    A \emph{generalized split schedule} $\schedule$ over $\transset$ is a multiversion schedule over $\transset$ that has the following form:
    $$\prefix{\trans[1]}{b_1}\cdot \trans[2]\cdot \ldots \cdot \trans[n] \cdot \postfix{\trans[1]}{b_1},$$
    where
    \begin{enumerate}
        \item \label{c:1} for each pair of transactions $\trans[i], \trans[j] \in \transset$ with either $j = i+1$ or $i = n$ and $j = 1$, there is a $b_i \in \trans[i]$ and $a_j \in \trans[j]$ such that $\dependson{b_i}{a_j}$;
        \item \label{c:2} if for some pair of transactions $\trans[i], \trans[j] \in \transset$ an operation $a_j \in \trans[j]$ depends on an operation $b_i \in \trans[i]$ in $\schedule$, then $j = i+1$, or $i = n$ and $j = 1$; and
        \item \label{c:3} each write operation occurring in $\schedule$ respects the commit order of $\schedule$.
    \end{enumerate}
\end{definition}

The conditions above correspond to the following. Condition (\ref{c:1}) says that there is a cyclic chain of dependencies; condition (\ref{c:2}) stipulates that this cycle is minimal; and, condition (\ref{c:3}) 
enforces that transactions writing to the same object install versions in the same order as they commit.

The next lemma says that a generalized split schedule functions as a counterexample schedule for conflict-robustness as well as for view-robustness. It forms a basic building block for Theorem~\ref{theo:condition}.
\begin{lemma}
    \label{lem:generalsplit}
    Let $\schedule$ be a generalized split schedule for a set of transactions $\transset$. Then, $\schedule$ is not conflict-serializable and not view-serializable.
\end{lemma}

\begin{toappendix}
    \subsection{Proof for Lemma~\ref{lem:generalsplit}}
\begin{proof}
    It immediately follows that $s$ is not conflict-serializable. Indeed, condition~(\ref{c:1}) stipulates that there is a dependency from each transaction $\trans[i]$ in $\schedule$ to the transaction $\trans[i+1]$, as well as a dependency from $\trans[n]$ to $\trans[1]$. This means in particular that there is a cycle in $\seg{\schedule}$. Hence, by Theorem~\ref{theo:not-conflict-serializable}, $\schedule$ is not conflict-serializable.

    The argument that $\schedule$ is not view-serializable is more involved. We argue that for each dependency $\dependson{b_i}{a_j}$ in $\schedule$ between some pair of transactions $\trans[i]$ and $\trans[j]$ in $\transset$, the transaction $\trans[i]$ must precede $\trans[j]$ in every single-version serial schedule $\schedule'$ that is view-equivalent to $\schedule$. Since Condition~(\ref{c:1}) implies a cycle of dependencies, such a single-version serial schedule $\schedule'$ view-equivalent to $\schedule$ cannot exist, thereby showing that $\schedule$ is not view-serializable.

    We argue for each type of dependency that $\trans[i]$ must precede $\trans[j]$ in every view-equivalent single-version serial schedule $\schedule'$. To this end, let $s'$ be a serial schedule view-equivalent to $s$. By Condition~(\ref{c:2}), we know that either $j = i+1$, or $i = n$ and $j = 1$. 

    \paragraph*{Case 1: $\dependson{b_i}{a_j}$ is a ww-dependency}
    Since each write operation in $\schedule$ respects the commit order in $\schedule$, we know that $\CT[i] \schords \CT[j]$. Hence, it can not be that $i = 1$ and $j = 2$. Furthermore, Condition~(\ref{c:2}) implies that a transaction $\trans[k]$ different from $\trans[i]$ and $\trans[j]$ and writing to the same object cannot exist, as such a write operation would introduce a ww-dependency.
    Indeed, assume towards a contradiction that a write operation $c_k \in \trans[k]$ exists, with $k \neq i$ and $k \neq j$.
    \begin{itemize}
        \item If $k \neq 1$ and $k < i$, then $\CT[k] \schords \CT[i] \schords \CT[j]$ and hence $c_k \schvord b_i \schvord a_j$ by Condition~(\ref{c:3}). The implied ww-dependency $\dependson{c_k}{a_j}$ then contradicts Condition~(\ref{c:2}). 
        \item Similarly, if $k = 1$ or $k > i$, then $\CT[i] \schords \CT[j] \schords \CT[k]$ and hence $b_i \schvord a_j \schvord c_k$ by Condition~(\ref{c:3}).
        Notice in particular that if $k > i$, then $k > j$ as well due to $j = i + 1$ and $k \neq j$.
        The implied ww-dependency $\dependson{b_i}{c_k}$ then contradicts Condition~(\ref{c:2}).
    \end{itemize}
    Since $\trans[i]$ and $\trans[j]$ are the only transactions writing to the object at hand, and since $b_i \schvord a_j$,
    we conclude that $\trans[j]$ must be the transaction in $\schedule$ writing the last version of this particular object.
    In the view-equivalent single-version serial schedule $\schedule'$, a transaction writing to this object (and $\trans[i]$ in particular) therefore cannot occur after $\trans[j]$.

    \paragraph*{Case 2: $\dependson{b_i}{a_j}$ is a wr-dependency}
    By definition of $\dependson{b_i}{a_j}$, either $b_i = \schvf(a_j)$, or $b_i \schvord \schvf(a_j)$. In the former case, we have $b_i = \schvf[\schedule'](a_j)$, since $\schedule$ is view-equivalent to $\schedule'$. Hence, $\trans[i]$ must precede $\trans[j]$ in $\schedule'$ as, by definition of a schedule, $\schvf[\schedule'](a_j) \schords[\schedule'] a_j$. 

    In the latter case (\ie, $b_i \schvord \schvf(a_j)$), assume that $c_k = \schvf(a_j)$ with $c_k \in \trans[k]$ for some transaction $\trans[k] \in \transset$. If $\trans[k] = \trans[i]$, the argument is analogous to the previous case where $b_i = \schvf(a_j)$. In particular, $c_k = \schvf(a_j)$ implies $c_k = \schvf[\schedule'](a_j)$ due to $\schedule$ being view-equivalent to $\schedule'$.
    If instead $\trans[k]$ is different from $\trans[i]$, then $b_i \schvord c_k$ implies a ww-dependency $\dependson{b_i}{c_k}$. By Condition~(\ref{c:2}) of Definition~\ref{def:generalizedsplitschedule}, this dependency implies $k=j$. Because of this ww-dependency from $\trans[i]$ to $\trans[j]$, we can now apply the same argument as in Case 1 to conclude that $\trans[i]$ must precede $\trans[j]$ in $\schedule'$.

    \paragraph*{Case 3: $\dependson{b_i}{a_j}$ is a rw-antidependency}
    Then, $\schvf(b_i) \schvord a_j$. We consider two cases, depending on whether another write operation on the same object as $a_j$ occurs in $\schedule$ writing a version installed after the version created by $a_j$.
    For the former case, assume no other write operation $c_k$ exists with $a_j \schvord c_k$.
    Then, the desired result is immediate, since $\trans[j]$ installs the last version of the object at hand in $\schedule$. Indeed, if $\trans[i]$ would occur after $\trans[j]$ in $\schedule'$, then $b_i$ would read the version installed by $\trans[j]$ in $\schedule'$
    which would imply that $\schvf[\schedule'](b_i) = a_j \neq \schvf[\schedule](b_i)$, thereby violating that $s'$ is view-equivalent with $s$.

    For the latter case, assume a write operation $c_k \in \trans[k]$ on the same object as $a_j$ exists, with $\trans[k] \in \transset$ different from $\trans[j]$ and $a_j \schvord c_k$. Then, there is a ww-dependency $\dependson{a_j}{c_k}$. By Condition~(\ref{c:2}), either $k = j+1$ or $j = n$ and $k = 1$ ($\dagger$). Furthermore, since $\schvord$ defines a total order for each set of write operations on the same object, we have $\schvf(b_i) \schvord a_j \schvord c_k$ and hence $\schvf(b_i) \schvord c_k$. So, unless $\trans[i] = \trans[k]$, $\dependson{b_i}{c_k}$ is a rw-antidependency as well. By the same argument as above, the antidependency $\dependson{b_i}{c_k}$ implies that either $k = i+1$ or $i = n$ and $k = 1$ ($\ddagger$). Combining ($\dagger$) and ($\ddagger$), we derive that $i = j$, contradicting the assumption that $\trans[i]$ and $\trans[j]$ are different transactions. We therefore conclude that $\trans[i] = \trans[k]$.

    Note that there is a rw-antidependency $\dependson{b_i}{a_j}$ from $\trans[i]$ to $\trans[j]$, and a ww-dependency $\dependson{a_j}{c_k}$ from $\trans[j]$ to $\trans[i]$.
    According to Condition~(\ref{c:2}), this can only occur if $n = 2$ (i.e., there are only two transactions in $\transset$).
    Since all write operations in $\schedule$ respect the commit order of $\schedule$ (Condition~(\ref{c:3})), and since $a_j \schvord c_k$, the structure of $\schedule$ (in particular, $\trans[2]$ committing before $\trans[1]$) implies that $\trans[i] = \trans[k] = \trans[1]$, and $\trans[j] = \trans[2]$.

    We now argue that $\trans[1]$ must precede $\trans[2]$ in every view-equivalent single-version serial schedule $\schedule'$. If $\schvf(b_i) = \sstart$, the argument is immediate, as $\trans[2]$ writes to the object at hand and therefore $\schvf[\schedule'](b_i) = \sstart$ can only occur when $\trans[1]$ precedes $\trans[2]$ in $\schedule'$. If instead $\schvf(b_i) = d_\ell$ for some write operation $d_\ell$ occurring in a transaction $\trans[\ell] \in \{\trans[1], \trans[2]\}$, the argument is based on whether $d_\ell \in \trans[1]$ or $d_\ell \in \trans[2]$. For both cases, we will argue that no view-equivalent single-version serial schedule $\schedule'$ can exist. If $d_\ell \in \trans[1]$, then $\schvf(b_i) = d_\ell \schvord a_j$ by definition of $\dependson{b_i}{a_j}$, contradicting our assumption that the write operations in $\schedule$ respect the commit order of $\schedule$ (recall that $\trans[1]$ commits after $\trans[2]$ in $\schedule$). If instead $d_\ell \in \trans[2]$, then $\schvf(b_i) = d_\ell \schvord a_j$ implies $d_\ell <_{\trans[2]} a_j$. But then $\schedule$ can never be view-equivalent to a single-version serial schedule $\schedule'$, since in such a schedule a read operation $b_i \in \trans[1]$ can never observe the version written by $d_\ell$ in $\trans[2]$. Indeed, if $\trans[1]$ precedes $\trans[2]$ in $\schedule'$, then $b_i$ cannot observe the version written by $d_\ell$ since $b_i \schords[\schedule'] d_\ell$, and if $\trans[2]$ precedes $\trans[1]$ in $\schedule'$, then $b_i$ cannot observe the version written by $d_\ell$, as this version is already overwritten by the more recent version written by $a_j$.
\end{proof}
\end{toappendix}

\subsection{Sufficient condition}

{We are now ready to formulate a sufficient condition for allocations for which conflict-robustness implies view-robustness:}

\condition{
    \label{cond:rob:equivalence}
    {
    An allocation $\alloc$ for a set of transactions $\transset$ satisfies Condition~\ref{cond:rob:equivalence} if a generalized split schedule $\schedule'$ over a set of transactions $\transset' \subseteq \transset$ such that $\schedule'$ is allowed under $\alloc[\transset']$ always exists when $\transset$ is not conflict-robust against $\alloc$.
    }
}

The next theorem shows that \ref{cond:rob:equivalence} is indeed a sufficient condition.
{
\begin{theorem}\label{theo:condition}
    Let $\alloc$ be an allocation for a set of transactions $\transset$ for which Condition~\ref{cond:rob:equivalence} holds. Then, the following are equivalent:
    \begin{enumerate}
        \item $\transset$ is conflict-robust against $\alloc$; and,
        \item $\transset$ is view-robust against $\alloc$.
    \end{enumerate}
\end{theorem}
}

\begin{toappendix}
    \subsection{Proof for Theorem~\ref{theo:condition}}
\begin{proof}
    $(1 \to 2)$ Follows from Proposition~\ref{prop:rob:basicproperty}(\ref{basic:one}).
    $(2 \to 1)$ The proof is by contraposition.
    Assume $\transset$ is not conflict-robust against $\alloc$. Then, by Condition~\ref{cond:rob:equivalence}, there exists a generalized split schedule $\schedule'$ over a set of transactions $\transset' \subseteq \transset$ such that $\schedule'$ is allowed under $\alloc[\transset']$. Furthermore, Lemma~\ref{lem:generalsplit} implies that $\schedule'$ is not view-serializable. Hence, $\transset$ is not view-robust against $\alloc$.
\end{proof}
\end{toappendix}

\subsection{Robustness against \mvrc, \si, and \ssi}
\label{sec:robustness-against-isolevels}

In this section, we obtain that for the isolation levels of PostgreSQL and Oracle, widening from conflict- to view-serializability does not allow for more sets of transactions to become robust.

\begin{theorem}\label{theo:rc-si-ssi-view-robust}
    Let $\alloc$ be an $\isopostgres$-allocation for a set of transactions $\transset$. Then, the following are equivalent:
    \begin{enumerate}
        \item $\transset$ is conflict-robust against $\alloc$;
        \item $\transset$ is view-robust against $\alloc$;
    \end{enumerate}
\end{theorem}

To prove the above theorem, it suffices to show by Theorem~\ref{theo:condition} that Condition~\ref{cond:rob:equivalence} holds for $\isopostgres$-allocations. It follows from Vandevoort et al.~\cite{DBLP:conf/pods/VandevoortKN23} that for $\isopostgres$-allocations, exact conflict-robustness is characterized in terms of the absence of schedules of a very specific form, referred to as \emph{multiversion split schedules}. In particular, they show that if a set of transactions $\transset$ is not exact conflict-robust against an $\isopostgres$-allocation $\alloc$, then there exists a multiversion split schedule $\schedule$ over $\transset$ that is allowed under $\alloc$ and that is not conflict-serializable. For a set of transactions $\transset = \{\trans[1], \trans[2], \ldots, \trans[n]\}$ with $n \geq 2$, a multiversion split schedule has the following form:
\[
    \prefix{\trans[1]}{b_1}\cdot \trans[2]\cdot \ldots \cdot \trans[m] \cdot \postfix{\trans[1]}{b_1} \cdot \trans[m+1] \cdot \ldots \cdot \trans[n],
\]
where $b_1 \in \trans[1]$ and $m \in [2, n]$.
Furthermore, for each pair of transactions $\trans[i], \trans[j] \in \transset$ with either $j = i+1$ or $i = m$ and $j = 1$, there is a $b_i \in \trans[i]$ and $a_j \in \trans[j]$ such that $\dependson{b_i}{a_j}$.
Although the structure of a multiversion split schedule is similar to the structure of a generalized split schedule, there are two important differences: (1) a multiversion split schedule allows a tail of serial transactions $\trans[m+1] \cdot \ldots \cdot \trans[n]$ to occur after $\postfix{\trans[1]}{b_1}$; and (2), the cyclic chain of dependencies is not necessarily minimal. Furthermore, Vandevoort et al.~\cite{DBLP:conf/pods/VandevoortKN23} only consider exact conflict-robustness, whereas we consider conflict-robustness. 
Our proof of Theorem~\ref{theo:rc-si-ssi-view-robust} follows the following steps:
(1) show that conflict-robustness and exact conflict-robustness coincides for $\isopostgres$-allocations; (2) show that a counterexample multiversion split schedule allowed under a $\isopostgres$-allocation $\alloc$ can always be transformed into a generalized split schedule that is allowed under $\alloc$. Theorem~\ref{theo:rc-si-ssi-view-robust} then readily follows.

From Theorem 3.3 in \cite{DBLP:conf/pods/VandevoortKN23} and Theorem~\ref{theo:rc-si-ssi-view-robust}, we have the following corollary:
\begin{corollary}
    Let $\alloc$ be an $\isopostgres$-allocation for a set of transactions $\transset$. 
    Deciding whether $\transset$ is view-robust against $\alloc$ is in polynomial time.
\end{corollary}

\begin{toappendix}
    \subsection{Proof for Theorem~\ref{theo:rc-si-ssi-view-robust}}

As mentioned in Section~\ref{sec:robustness:definitions}, conflict-robustness and exact conflict-robustness are equivalent since a non-conflict-serializable schedule $\schedule'$ over a subset of transactions $\transset' \subseteq \transset$ can always be extended to a non-conflict-serializable $\schedule$ over $\transset$, under the assumption that appending transactions in a serial fashion does not violate the corresponding isolation level requirements. The next Lemma provides a more formal argument for $\isopostgres$-allocations.

\begin{lemma}
    \label{lemma:nonconfser:extended}
    Let $\alloc$ be an $\isopostgres$-allocation for a set of transactions $\transset$, and let $\schedule'$ be a schedule over a subset $\transset' \subseteq \transset$ that is allowed under $\alloc[\transset']$ and is not conflict-serializable. Then, there exists a schedule $\schedule$ over $\transset$ that is allowed under $\alloc$ and is not conflict-serializable.
\end{lemma}

\begin{proof}
    The schedule $\schedule$ is constructed from $\schedule'$ by extending the schedule with the remaining transactions in $\transset$ (i.e., those not appearing in $\transset'$) in a serial fashion.
    More formally, we extend $\schords[\schedule']$ towards $\schords$ by taking $b_i \schords a_j$ for each pair of operations $b_i \in \trans[i]$ and $a_j \in \trans[j]$ with $\trans[i] \in \transset'$ and $\trans[j] \not\in \transset'$. Furthermore, for each pair of operations $a_i, b_i \in \trans[i]$ with $\trans[i]$ a transaction not appearing in $\transset'$, there is no operation $c_j \in \trans[j]$ with $\trans[i] \neq \trans[j]$ such that $a_i \schords c_j \schords b_i$.
    The version order $\schvord[\schedule']$ is extended towards $\schvord$ by taking $b_i \schvord a_j$ for each pair of write operations $b_i \in \trans[i]$ and $a_j \in \trans[j]$ with $\trans[i] \in \transset'$ and $\trans[j] \not\in \transset'$. For each pair of write operations $b_i$ and $a_j$  with $\trans[i], \trans[j] \not\in \transset'$, we take $b_i \schvord a_j$ iff $b_i \schords a_j$.
    Finally, for each read operation $b_i$ not occurring in $\transset'$ the version function $\schvf(b_i)$ is defined such that $b_i$ observes the last written version. That is, there is no write operation $a_j$ on the same object as $b_i$ with $\schvf(b_i) \schords a_j \schords b_i$.
    
    We argue that $\schedule$ is indeed allowed under $\alloc$.
    By construction of $\schedule$ based on extending $\schedule'$, a transaction in $\transset'$ allowed under \rc (respectively \si) in $\schedule'$ is allowed under \rc (respectively \si) in $\schedule$ as well. Since transactions not in $\transset'$ are appended serially in $\schedule$, they are allowed under \si in $\schedule$. Furthermore, since the newly added transactions are never concurrent with other transactions, they cannot introduce additional dangerous structures in $\schedule$. Hence, since $\schedule'$ is allowed under $\alloc[\transset']$, we conclude that $\schedule$ is allowed under $\alloc$.

    It remains to argue that $\schedule$ is not conflict-serializable. To this end, note that by construction of $\schedule$ the order of operations $\schords$, version order $\schvord$ and version function $\schvf$ are identical to respectively $\schords[\schedule']$, $\schvord[\schedule']$ and $\schvf[\schedule']$ when restricting our attention to operations in $\transset'$. Because of this, each dependency $\dependson[\schedule']{b_i}{a_j}$ in $\schedule'$ implies an equivalent dependency $\dependson[\schedule]{b_i}{a_j}$ in $\schedule$. A cycle in $\seg{\schedule'}$ therefore implies a cycle in $\seg{\schedule}$ as well. Since $\schedule'$ is not conflict-serializable, $\schedule$ is not conflict-serializable either.
\end{proof}

Towards the proof of Theorem~\ref{theo:rc-si-ssi-view-robust}, 
the next Lemma states that for a non-conflict-serializable schedule $\schedule$ witnessed by a cycle $C$ in the $\seg{\schedule}$, transactions not in $C$ can be removed while maintaining a non-conflict-serializable schedule.  

\begin{lemma}
    \label{lemma:nonconfser:minimal}
    Let $\schedule$ be a non-conflict-serializable schedule over a set of transactions $\transset$ allowed under an $\isopostgres$-allocation $\alloc$, and let $C$ be a cycle in $\seg{\schedule}$ with $\transset' \subseteq \transset$ the set of transactions in $C$. Then, there exists a non-conflict-serializable schedule $\schedule'$ over $\transset'$ allowed under $\alloc[\transset']$.
\end{lemma}

\begin{proof}
    Let $\schedule'$ be the schedule obtained from $\schedule$ by removing all operations from transactions not in $\transset'$. For the order of operations $\schords[\schedule']$ and version order $\schvord[\schedule']$ this construction is trivial, but the version function $\schvf[\schedule']$ requires more attention. Indeed, for a read operation $b_i$ in $\schedule'$, the version $\schvf(b_i)$ observed in $\schedule$ is not necessarily an operation in $\transset'$. In this case, we define $\schvf[\schedule'](b_i)$ such that it observes the last version written in $s'$. 
    More formally, for each read operation $b_i \in \trans[i]$ in $\schedule'$, we require that $\schvf[\schedule'](b_i) = \schvf(b_i)$ if $\schvf(b_i)$ is an operation in $\schedule'$.
    Otherwise, $\schvf[\schedule'](b_i)$ is the last operation occurring in $\transset'$ before $\schvf(b_i)$ according to $\schvord$. That is, $\schvf[\schedule'](b_i) \schvord \schvf[\schedule](b_i)$ and there is no other write operation $a_j$ occurring in $\transset'$ on the same object with $\schvf[\schedule'](b_i) \schvord a_j \schvord \schvf[\schedule](b_i)$.

    By construction, $\schedule'$ cannot introduce additional dirty writes, concurrent writes or dangerous structures. Furthermore, since all write operations in $\schedule$ respect the commit order of $\schedule$, the same holds for $\schedule'$. To argue that $\schedule'$ is allowed under $\alloc[\transset']$, it therefore suffices to show that if a read operation $b_j \in \trans[j]$ occurring in $\transset'$ is read-last-committed in $\schedule$ relative to an operation $a_j \in \trans[j]$, then $b_j$ is read-last-committed in $\schedule'$ relative to $a_j$ as well.
    Since either $\schvf[\schedule'](b_j) = \schvf[\schedule](b_j)$ or $\schvf[\schedule'](b_j) \schvord \schvf[\schedule](b_j)$ and since all write operations in $\schedule$ respect the commit order of $\schedule'$, we conclude that either $\schvf[\schedule'](b_j) = \sstart$ or $\CT[i] \schords[\schedule'] a_j$ with $\schvf[\schedule'](b_j) \in \trans[i]$.
    Towards a contradiction, assume a write operation $c_k \in \trans[k]$ occurring in $\transset'$ over the same object as $b_j$ exists with $\CT[k] \schords[\schedule'] a_j$ and $\schvf[\schedule'](b_j) \schvord[\schedule'] c_k$. By construction of $\schedule'$ based on $\schedule$, we have $\CT[k] \schords[\schedule] a_j$ and $\schvf[\schedule'](b_j) \schvord c_k$.
    Since $b_j$ is read-last-committed in $\schedule$ relative to $a_j$, either $\schvf[\schedule](b_j) = c_k$ or $c_k \schvord \schvf[\schedule](b_j)$. Due to $\schvf[\schedule'](b_j) \schvord[\schedule'] c_k$, both options contradict our choice of $\schvf[\schedule'](b_j)$ in the construction of $\schedule'$.

    It remains to show that $\schedule'$ is not conflict-serializable. To this end, let $b_i \in \trans[i]$ and $a_j \in \trans[j]$ be two operations with $\trans[i], \trans[j] \in \transset'$ such that $\dependson[\schedule]{b_i}{a_j}$. We argue that $\dependson[\schedule']{b_i}{a_j}$ is an (anti)dependency in $\schedule'$ as well. If $\dependson[\schedule]{b_i}{a_j}$ is a ww-dependency, the result is immediate, since $b_i \schvord a_j$ implies $b_i \schvord[\schedule'] a_j$.
    If $\dependson[\schedule]{b_i}{a_j}$ is a rw-antidependency, then $\schvf(b_i) \schvord a_j$. By construction of $\schedule'$, either $\schvf[\schedule'](b_i) = \schvf(b_i)$ or $\schvf[\schedule'](b_i) \schvord \schvf(b_i)$. Consequently, $\schvf[\schedule'](b_i) \schvord a_j$ and hence $\schvf[\schedule'](b_i) \schvord[\schedule'] a_j$, thereby witnessing $\dependson[\schedule']{b_i}{a_j}$.
    Finally, if $\dependson[\schedule]{b_i}{a_j}$ is a wr-dependency, then $b_i = \schvf(a_j)$ or $b_i \schvord \schvf(a_j)$. In the former case, $b_i = \schvf[\schedule'](a_j)$, since $\schvf[\schedule'](a_j) = \schvf(a_j)$ by construction of $\schedule'$. In the latter case, it can be argued that either $b_i = \schvf[\schedule'](a_j)$ or $b_i \schvord[\schedule'] \schvf[\schedule'](a_j)$ since all write operations in $\schedule$ and $\schedule'$ respect the commit order in respectively $\schedule$ and $\schedule'$, and since $a_j$ is read-last-committed relative to some operation $c_j \in \trans[j]$ in both $\schedule$ and $\schedule'$. We conclude that in both cases $\dependson[\schedule']{b_i}{a_j}$ is implied.
    
    Because each (anti)dependency in $\schedule$ between operations occurring in $\transset'$ is preserved in $\schedule'$, it is now easy to see that the cycle $C$ in $\seg{\schedule}$ implies a cycle $\seg{\schedule'}$ as well, thereby witnessing that $\schedule'$ is not conflict-serializable.
\end{proof}

For an initial set of transactions $\transset$, the proof of Theorem~\ref{theo:rc-si-ssi-view-robust} now relies on a repeated application of Lemma~\ref{lemma:nonconfser:minimal} until we obtain a non-conflict-robust set of transactions $\transset' \subseteq \transset$ for which the implied multiversion split schedule $\schedule'$ is a generalized split schedule. Indeed, if $\schedule'$ is not a generalized split schedule, then we can always remove at least one transaction from $\transset'$ by picking a minimal cycle $C$ in $\seg{\schedule'}$ and applying Lemma~\ref{lemma:nonconfser:minimal} to obtain a non-conflict-serializable schedule $\schedule''$ over a set of transactions $\transset'' \subsetneq \transset'$. Since $\schedule''$ is allowed under $\alloc[\transset'']$, this schedule witnesses the fact that $\transset''$ is not conflict-robust against $\alloc[\transset'']$. Since each iteration removes at least one transaction, the algorithm terminates after a finite number of steps.

\begin{proof}[Proof of Theorem~\ref{theo:rc-si-ssi-view-robust}]
    To prove Theorem~\ref{theo:rc-si-ssi-view-robust} it suffices to show by Theorem~\ref{theo:condition} that Condition~\ref{cond:rob:equivalence} holds for $\isopostgres$-allocations.
    To this end, let $\transset$ be a set of transactions not conflict-robust against an $\isopostgres$-allocation $\alloc$ for $\transset$.
    This means that there exists a schedule $\schedule'$ over a subset $\transset' \subseteq \transset$ that is allowed under $\alloc[\transset']$ and not conflict-serializable. By Lemma~\ref{lemma:nonconfser:extended}, we can extend $\schedule'$ to a schedule over $\transset$ that is allowed under $\alloc$ and not conflict-serializable, thereby witnessing $\transset$ being not exact conflict-robust against $\alloc$.
    By Theorem 3.3 in \cite{DBLP:conf/pods/VandevoortKN23}, a multiversion split schedule $\schedule$ over $\transset$ exists that is allowed under $\alloc$ and not conflict-serializable.
    Recall that the multiversion split schedule $\schedule$ has the following form:
    \[
        \prefix{\trans[1]}{b_1}\cdot \trans[2]\cdot \ldots \cdot \trans[m] \cdot \postfix{\trans[1]}{b_1} \cdot \trans[m+1] \cdot \ldots \cdot \trans[n],
    \]
    where $b_1 \in \trans[1]$ and $m \in [2, n]$.
    If $\schedule$ is a generalized split schedule as well, the result is immediate. Otherwise, $\schedule$ either has a nonempty tail of transactions (i.e., $m < n$), or the cycle $C$ in $\seg{\schedule}$ corresponding to the cyclic chain of dependencies between each $\trans[i]$ and $\trans[i+1]$ with $i \in [1, m-1]$, as well as between $\trans[m]$ and $\trans[1]$ is not minimal. In both cases, a (minimal) cycle $C'$ in $\seg{\schedule}$ exists that does not involve all transactions in $\transset$. By Lemma~\ref{lemma:nonconfser:minimal}, we can remove transactions not in $C'$ while maintaining a non-conflict-serializable schedule $\schedule'$ over a set of transactions $\transset' \subsetneq \transset$. Since $\schedule'$ is allowed under $\alloc[\transset']$, this schedule witnesses the fact that $\transset'$ is not conflict-robust against $\alloc[\transset']$.
    Since $\transset'$ is not conflict-robust against $\alloc[\transset']$, a multiversion split schedule $\schedule''$ over $\transset'$ exists that is allowed under $\alloc[\transset']$ and not conflict-serializable.
    We can now repeat the argument above for $\schedule''$ instead of $\schedule$, until the obtained multiversion split schedule is a generalized split schedule as well.
    Because each iteration removes at least one transaction from the remaining set of transactions, the algorithm terminates after a finite number of steps.
\end{proof}
\end{toappendix}


\section{Deciding View-Serializability}
\label{sec:deciding:viewser}

It is well-known that view-serializability is NP-hard~\cite{DBLP:books/cs/Papadimitriou86}. The proof is based on a reduction from the polygraph acyclicity problem but the resulting schedules make extensive use of dirty writes which are not permitted under \mvrc or \si. We present a modified reduction that avoids dirty writes and obtain that 
deciding view-serializability remains NP-hard even when the input is restricted to schedules only consisting of transactions that are allowed under RC, or SI, respectively.

\begin{theorem}\label{thm:vs:rc:si:np:hard}
Deciding view-serializability is NP-hard, even for schedules only consisting of transactions allowed under RC, or SI, respectively.
\end{theorem}

\begin{toappendix}
    \subsection{Proof for Theorem~\ref{thm:vs:rc:si:np:hard}}

\begin{proof}
We first recall the essential concepts relating to polygraphs:
A polygraph is a triple $P=(V, A, C)$ consisting of a finite set $V$ of nodes, a set $A\subseteq V^2$ of arcs, and a set $C\subseteq V^3$ of choices.
For choices $(u,v,w) \in C$ it is assumed that all nodes are distinct and that there is an arc $(w,u)\in A$.
A polygraph $P=(V, A, C)$ is called \emph{compatible} with a directed graph $(V, B)$ if $A\subseteq B$ and for each choice $(u,v,w) \in C$ either $(u,v) \in B$ or $(v,w) \in B$. Finally, a polygraph $P$ is said to be \emph{acyclic} if a directed acyclic graph exists that is compatible with $P$.

\paragraph{Reduction.}

Given a polygraph $P=(V,A,C)$, we will now construct a schedule $\schedule$ over a specific set $\transset_P$ of transactions with property that $\schedule$ is view-serializable if and only if $P$ is acyclic. For the construction, we assume existence of an object $\x_{(w,u)}$ for every arc $(w,u)\in A$ as well as an object
$\x_{(u,v,w)}$ for every choice $(u,v,w) \in C$. All these objects are assumed to be unique. The set $\transset_P$ consists of a transaction
$\trans[w]$ for every node $w\in V$
and two transactions, $\trans[(u,v,w), 0]$ and $\trans[(u,v,w),\infty]$, for every choice $(u,v,w) \in C$. Next, we will specify the operations in these transactions. As the order of operations in the transactions themselves is unimportant for the construction, we will only specify which operations are present (besides the mandatory commit operation that always occurs at the end of the transaction).

Formally, transactions $\trans[w]$ have the following operations:
\begin{itemize}
\item a read operation on $\x_{(w,u)}$ for every $(w,u) \in A$;
\item a write operation on $\x_{(u,w)}$ for every $(u,w) \in A$;
\item a read operation on $\x_{(u,v,w)}$ for every $(u,v,w) \in C$;
\item a write operation on $\x_{(u,w,v)}$ for every $(u,w,v) \in C$; and
\item a read operation on $\x_{(w,v,u)}$ for every $(w,v,u) \in C$.
\end{itemize}

Transactions $\trans[(u,v,w), 0]$ and $\trans[(u,v,w), \infty]$ both consist only of a write operation on $\x_{(u,v,w)}$.

We will now construct the target schedule $\schedule = (\schop, \schord, \schvord, \schvf)$ over $\transset_P$. This schedule can be seen as a concatenation of three different schedules: a single-version schedule over $\{\trans[(u,v,w), 0] \mid (u,v,w) \in C\}$; a schedule over $\{\trans[w]\mid w\in V\}$; and a single-version schedule over $\{\trans[(u,v,w), \infty] \mid (u,v,w) \in C\}$. The order of  transactions in the first and last part is unimportant.
In the middle part, all transactions are concurrent. For a formal definition of $\schord$ for the middle part, we can choose an arbitrary order $\le_V$ on the transactions $\{\trans[w]\mid w\in V\}$ and let $\schord$ consist of all the first operations of transactions in $\{\trans[w]\mid w\in V\}$ in the order decided by $\le_V$, then in an arbitrary order all normal operations of these transactions (except the commit), and finally the commit operations of these transactions, again in the order decided by $\le_V$.
Schematically, we get the pattern below in $\schedule$ for each arc $(w,u) \in A$:
\[
\begin{array}{rc}
\trans[w]: & \ldots \R[]{\x_{(w,u)}} \ldots \CT[] \\    
\trans[u]: & \ldots \W[]{\x_{(w,u)}} \ldots \CT[] \\
\end{array}
\]
and the below pattern for each choice $(u,v,w) \in C$:
\[
\begin{array}{rccc}
\trans[(u,v,w), 0]: & \W[]{\x_{(u,v,w)}} \, \CT[] \\
\trans[u]: & & \ldots \R[]{\x_{(u,v,w)}} \ldots \CT[] \\
\trans[v]: & & \ldots \W[]{\x_{(u,v,w)}} \ldots \CT[] \\
\trans[w]: & & \ldots \R[]{\x_{(u,v,w)}} \ldots \CT[] \\
\trans[(u,v,w), \infty]: & & & \W[]{\x_{(u,v,w)}} \, \CT[] \\
\end{array}
\]

We remark that
\begin{align}
    \parbox[t]{.8\linewidth}{$\schedule$ has no concurrent writes.}\label{cond:concurrentwrites}
\end{align}
Indeed, only the middle part has concurrent transactions, and no two of these transactions write to a same object.

We choose the version order $\schvord$ consistent with $\schord$. For this, we notice that only object $\x_{(u,v,w)}$ is written to by multiple operations, namely a write operation $o_1$ by $\trans[(u,v,w), 0]$, a write operation $o_2$ by $\trans[v]$, and a write operation $o_3$ by $\trans[(u,v,w), \infty]$.
Since $\trans[(u,v,w), 0]$, $\trans[v]$, and $\trans[(u,v,w), \infty]$ occur in respectively the first, middle and last part of $\schedule$, it follows directly from the construction of $\schedule$ that
\begin{align}
    \parbox[t]{.8\linewidth}{all operations in $\schedule$ respect the commit order of $\schedule$.}\label{cond:commitorder}
\end{align}

Finally, we choose the version function $\schvf$ to map all read operations on the last commit version before itself, hence we trivially obtain that
\begin{align}
    \parbox[t]{.8\linewidth}{each read operation in $\schedule$ is read-last-committed in $\schedule$ relative to itself.}\label{cond:readlastcommitteda}
\end{align}
Since the read operations are all occurring in the middle part of $\schedule$, the above statement is equivalent to
\begin{align}
    \parbox[t]{.8\linewidth}{each read operation in $\schedule$ is read-last-committed in $\schedule$ relative to the first operation of the transaction it belongs to.}\label{cond:readlastcommittedb}
\end{align}
Indeed, all read operations occur in the middle part of $\schedule$ and are over operations $\x_{(w,u)}$ and $\x_{(u,v,w)}$. For read operations on objects $\x_{(w,u)}$ we have $\schvf(\x_{(w,u)}) = \sstart$ and for read operations on objects $\x_{(u,v,w)}$ we have that $\schvf(\x_{(u,v,w)})$ is the unique write operation on $\x_{(u,v,w)}$ in transaction $\trans[(u,v,w), 0]$, which is part of the first part of $\schedule$.

We have now constructed a schedule $\schedule$ over $\transset_P$ from polygraph $P$ whose transactions are clearly allowed under both RC and SI, as follows directly from conditions \ref{cond:concurrentwrites}, \ref{cond:commitorder}, \ref{cond:readlastcommitteda}, and \ref{cond:readlastcommittedb}. We remark that our construction is also clearly a polynomial time construction.

Next, we show correctness of $\schedule$. That is, we show that $P$ is acyclic if and only if $\schedule$ is view-serializable.

\paragraph{If $P$ is acyclic then $\schedule$ is view-serializable.}
Since $P=(V,A,C)$ is acyclic, there exists an acyclic directed graph $G=(V,B)$ that is compatible with $P$.
Let $\le_V$ be a topological order of this graph $P$. We remark that $\le_V$ indeed always exists for acyclic directed graphs.
Next, we will construct a single-version serial schedule $\schedule'$ of $\transset_P$ and show that it is view-equivalent to $\schedule$.
As can be expected, the choice of $\schord[\schedule']$ for transactions $\trans[v]$ (with $v\in V$) will coincide with $\le_V$. But we recall that the set $\transset_P$ also contains transactions $\trans[(u,v,w), 0]$ and $\trans[(u,v,w), \infty]$ for choices $(u,v,w)\in C$ whose place we still need to decide.
The transactions $\trans[(u,v,w), \infty]$ occur at the end of the schedule in some arbitrary order (for example the same order as in $\schedule$). For every $v\in V$ consider sets $\transset_v \mydef \{ \trans[(u,v,w), 0] \mid (u,v,w)\in C \}$. We position all transactions of such a set $\transset_v$ directly in front of transaction $\trans[v]$. The order of transactions in $\transset_v$ is again not important for the construction. 

It is clear that $\schedule'$ is a serial schedule, and we choose $\schvord[\schedule']$ and $\schvf[\schedule']$ such that $\schedule'$ is a single-version schedule.
It remains to show that $\schedule'$ is also view-equivalent to $\schedule$. More precisely, we still need to argue that (i) $\lastv[\schedule]{\x} = \lastv[\schedule']{\x}$ for every object $\x$ and that (ii) $\schvf[\schedule](o) = \schvf[\schedule'](o)$ for every read operation $o$.
For objects $\x_{(w, u)}$  condition (i) is immediate, as there is only one write operation on $\x_{(w,u)}$ (which occurs in transaction $\trans[u]$).
For objects $\x_{(u,v,w)}$, $\lastv[\schedule]{\x_{(u,v,w)}}$ and $\lastv[\schedule']{\x_{(u,v,w)}}$ clearly are the operation on $\x_{(u,v,w)}$ in transaction $\trans[(u,v,w), \infty]$.

For condition (ii), we remark that for objects $\x_{(w,u)}$, the only read operation $o$ on it is in transaction $\trans[w]$, and the only write operation $o'$ on it is in transaction $\trans[u]$.
Since 
    $\trans[w]\schord[\schedule'] \trans[u]$ it is immediate that $\schvf[\schedule](o) = \sstart = \schvf[\schedule'](o)$. For objects $\x_{(u,v,w)}$ the argument is slightly more complicated. In schedule $\schedule$ we have that both read operations 
$o$ and $o'$ on $\x_{(u,v,w)}$ in $\trans[u]$ and $\trans[w]$, respectively, read the version written by the write $\schvf[\schedule](o) = \schvf[\schedule](o') = o''$ on this object in transaction $\trans[(u,v,w), 0]$. Particularly notice that the other transactions writing to $\x_{(u,v,w)}$, being transaction $\trans[v]$ and $\trans[(u,v,w), \infty]$, have their write after the start of transactions $\trans[u]$ and $\trans[w]$.

For schedule $\schedule'$, we remark that $\trans[(u,v,w), 0]$ is placed before $\trans[w]$, thus with $\schvf[\schedule'](o') = o''$, and that $\trans[w] \schord[\schedule'] \trans[u]$ and either $\trans[u] \schord[\schedule'] \trans[v]$ or  $\trans[v] \schord[\schedule'] \trans[w]$ implying $\schvf[\schedule'](o) = o''$ as desired. 

We conclude that $\schedule'$ is view-equivalent to $\schedule$ and thus that $\schedule$ is indeed view-serializable.

\paragraph{If $\schedule$ is view-serializable then $P$ is acyclic.}

If $\schedule$ is view-serializable, then there is a schedule $\schedule'$ over $\transset_P$ that is view-equivalent with $\schedule$. 
Let $G=(V, B)$ be the graph obtained by adding for every pair of nodes $(w,u) \in V^2$ the arc $(w,u)$ to $B$ iff $\trans[w] <_{\schedule'} \trans[u]$.
We remark that $P$ defines a total order on $V$ and is thus clearly acyclic. It remains to show that the arcs $B$ are consistent with arcs $A$ and choices $C$ of polygraph $P$.

For this, we first consider an arbitrary arc $(w,u) \in A$ from $P$. The latter implies that we constructed $\trans[w]$ to have a read operation $o$ on object $\x_{(w,u)}$ with $\schvf[\schedule](o) = \sstart$. Due to view-equivalence between $\schedule$ and $\schedule'$, we also have $\schvf[\schedule'](o) = \sstart$.
Now since $\trans[u]$ has been constructed to have a write on $\x_{(w,u)}$ it follows from $\schedule$ being single-version that $\trans[w] \schord[\schedule']\trans[u]$ and thus $(w,u) \in B$. We remark that $(w,u) \in B$ also implies $(u,w) \not\in B$ since $\schedule$ is single-version.

Now consider an arbitrary choice $(u,v,w) \in C$ from $P$. The latter implies $(w,u) \in A$ and thus we already know from the previous argument that $\trans[w] \schord[\schedule']\trans[u]$. It follows from our construction that both $\trans[w]$ and $\trans[u]$ have a read operation $o$ and $o'$, respectively, on objects $\x_{(u,v,w)}$ with $\schvf[\schedule](o) =\schvf[\schedule](o') = o''$ with $o''$ the write operation on $\x_{(u,v,w)}$ in $\trans[(u,v,w), 0]$.
View-equivalence between $\schedule$ and $\schedule'$ implies $\schvf[\schedule'](o) =\schvf[\schedule'](o') = o''$ and thus $\trans[(u,v,w), 0] \schord[\schedule'] \trans[w] \schord[\schedule'] \trans[u]$ since $\schedule'$ is single-version. 

Now it follows from the fact that we constructed $\trans[v]$ to have a write operation on that same object $\x_{(u,v,w)}$ that either $\trans[v] \schord[\schedule'] \trans[(u,v,w), 0] \schord[\schedule'] \trans[w]$ or $\trans[u] \schord[\schedule'] \trans[v]$, hence either $(v,w) \in B$ or $(u,v) \in B$, as desired.

We conclude that $B$ is indeed compatible with $P$ and hence that $P$ is acyclic 
\end{proof}
\end{toappendix}

We stress that Example~\ref{ex:vs:not:cs} and Example~\ref{ex:schedule} already show that view- and conflict-serializability do not coincide for the class of isolation levels $\{\mvrc,\si,\ssi\}$.


\section{Related Work}
\label{sec:relatedwork}

\subsection{View-serializability}

For single-version schedules, Yannakakis~\cite{DBLP:journals/jacm/Yannakakis84} formally defines view-serializability. For this problem, NP-hardness follows from a trivial extension based on a result by Papadimitriou~\cite{DBLP:journals/jacm/Papadimitriou79b}, proving that deciding final-state serializability is NP-complete.
Yannakakis~\cite{DBLP:journals/jacm/Yannakakis84} furthermore proves that within this setting of single-version schedules, the class of conflict-serializable schedules is the largest monotonic class contained in the class of view-serializable schedules. That is, a single-version schedule $\schedule$ is conflict-serializable if and only if all its (single-version) subschedules (including $\schedule$ itself) are view-serializable. 

Bernstein and Goodman~\cite{DBLP:journals/tods/BernsteinG83} extend concurrency control towards multiversion databases. Their definition of view-equivalence only requires the two schedules to observe identical versions for each read operation, and does not require the two schedules to install the same last versions for each object (as we do in this paper). The rationale is that when each write operation introduces a new version, the final database state will contain all versions and subsequent transactions executed afterwards could access any version that is required. However, in the context of isolation levels, a context not considered in \cite{DBLP:journals/tods/BernsteinG83}, restrictions can be put on which versions can be read by subsequent transactions. For instance, practical systems like PostgreSQL always require that the last committed version should be read (which allows the DBMS to safely remove old versions when all active transactions can only observe newer versions of this object). In the setting of this paper, it therefore makes sense to additionally require that view-equivalent schedules install the same last versions for each object.

Our definition of view-serializability is furthermore different from multiversion serializability studied by Papadimitriou and Kanellakis~\cite{DBLP:journals/tods/PapadimitriouK84}. In particular, our definition assumes the version function $\schvf$ to be fixed while searching for a view-equivalent single-version serial schedule, whereas their setting only assumes the order of operations $\schord$ to be fixed. Such an order is then said to be multiversion serializable if a version function exists such that the resulting schedule is view-equivalent to a single-version serial schedule. 
For this notion of multiversion serializability, they additionally consider a setting where at any point in time only the $k$ most recent versions of each object are stored for some fixed positive integer $k$. It is interesting to note that the isolation levels considered in this paper do not imply a fixed upper bound on the number of versions for each object that must be stored at any point in time. Indeed, since transactions allowed under \si and \ssi essentially take a snapshot of the versions visible at the start of the transaction, concurrent transactions can require a potentially unbounded number of versions to be maintained.
Finally, in the context of isolation levels that effectively constrain version functions, as considered in this paper, it is sensible to consider the version function to be fixed when defining serializability.

\subsection{Robustness and allocation for transactions}
    Yannakakis~\cite{DBLP:journals/jacm/Yannakakis84} studies \emph{$S$-safety} of transaction sets with respect to different notions of serializability $S$. A set of transactions $\transset$ is said to be \emph{$S$-safe} if every (single-version) schedule $\schedule$ over a subset of $\transset$ is $S$-serializable. This notion of $S$-safety therefore corresponds to $S$-robustness against allocations allowing \emph{all} single-version schedules. A particularly interesting result by Yannakakis~\cite{DBLP:journals/jacm/Yannakakis84} is that view-safety and conflict-safety coincide. That is, for a given set of transactions $\transset$, one can construct a single-version schedule that is not view-serializable over a subset of $\transset$ iff a single-version schedule over a subset of $\transset$ exists that is not conflict-serializable.
    Our work can therefore be seen as a generalization of this result to specific classes of isolation levels. It is worth to point out that Proposition~\ref{prop:rob:basicproperty}(\ref{basic:two}) already shows that the property does not hold for all classes of isolation levels.

Fekete~\cite{DBLP:conf/pods/Fekete05} is the first work that provides a necessary and sufficient condition for deciding {exact conflict-robustness} against an isolation level (\si) for a workload of transactions. In particular, that work provides a characterization for optimal allocations when every transaction runs under either \snapshot
or strict two-phase locking (S2PL).
{As a side result, a characterization for exact conflict-robustness against \snapshot is obtained.}
Ketsman et al.~\cite{DBLP:conf/pods/Ketsman0NV20} provide characterisations for exact conflict-robustness against \readcom and \readun under lock-based semantics. In addition, it is shown that the corresponding decision problems are complete for \coNP and \LOGSPACE, respectively, which should be contrasted with the polynomial time characterization obtained in \cite{vldbpaper} for exact conflict-robustness against {\it multiversion} read committed which is the variant that is considered in this paper.
{Vandevoort et al.~\cite{DBLP:conf/pods/VandevoortKN23} 
consider exact conflict-robustness against allocations over \isopostgres for which they provide a polynomial time decision procedure. They show in addition that there always is a unique optimal robust allocation over \isopostgres and provide a polynomial time algorithm to compute it.} {Other work studies exact conflict-robustness within a framework for uniformly specifying different isolation levels in a declarative way~\cite{DBLP:conf/concur/Cerone0G15,DBLP:conf/concur/0002G16,Cerone:2018:ASI:3184466.3152396,cerone_et_al:LIPIcs:2017:7794}. A key assumption here is \emph{atomic visibility} requiring that either all or none of the updates of each transaction are visible to other transactions. These approaches aim at higher isolation levels and cannot be used for \MVRC, as \MVRC does not admit \emph{atomic visibility}. None of these works consider view-robustness.}

Finally, we mention that the robustness problem is orthogonal to the problem of deciding whether a schedule is allowed under an isolation level. Biswas and Enea~\cite{DBLP:journals/pacmpl/BiswasE19} study this problem for a setting where schedules are represented by a partial order over the transactions (referred to as the session order) and a write-read relation, indicating for each read operation the write operation that wrote the observed version. A schedule is allowed under an isolation level if the partial order can be extended to a total order over the transactions consistent with the write-read relation as well as the criteria specific to the isolation level. For this setting, they show that verifying \textsc{read committed}, \textsc{read atomic} and \textsc{causal consistency} can be done in polynomial time, but deciding \textsc{prefix consistency} and \textsc{snapshot isolation} are NP-complete. Their setting should be contrasted with ours, where schedules are given as a total order over all operations and with a fixed version function and version order.


\section{Conclusions}
\label{sec:conclusions}

We showed that conflict- and view-robustness coincide for the class of isolation levels \isopostgres. The main implication is that for systems deploying these isolation levels, widening from conflict- to view-serializability does not allow for more sets of transactions to become robust. In addition, this paper can be seen as an extension of work by Yannakakis~\cite{DBLP:journals/jacm/Yannakakis84} who studied robustness (under the name of safety) for various serializability notions and obtained that view- and conflict-robustness coincide when \emph{all} schedules are allowed. 

An interesting direction for future work is studying robustness relative to other notions of serializability. A particularly interesting notion is the one of strict view-serializability. This notion extends view-serializability by requiring that the relative order of non-concurrent transactions should be preserved in the equivalent single-version serial schedule. Trivially, strict view-robustness implies view-robustness as every strict view-serializable schedule is view-serializable as well. However, since conflict-serializable schedules are not always strict view-serializable, conflict-robustness does not imply strict view-robustness in general. In fact, it can be shown that this result remains to hold even if we restrict our attention to $\{\rc\}$-allocations or $\{\si\}$-allocations,
and the complexity of strict view-robustness is therefore still open for $\isopostgres$-allocations. It would furthermore be interesting to identify relevant subsets where conflict-robustness coincides with strict view-robustness. 


\begin{acks}
    This work is partly funded by FWO-grant G019921N.
\end{acks}

\bibliographystyle{ACM-Reference-Format}
\bibliography{references}



\end{document}